\documentclass[11pt,a4paper,onecolumn,unpublished]{quantumarticle}

\pdfoutput=1

\usepackage[T1]{fontenc} 
\usepackage[english]{babel} 
\usepackage[utf8]{inputenc}
\usepackage{csquotes}

\usepackage[table,dvipsnames]{xcolor}
\usepackage[psdextra]{hyperref}
\usepackage{bookmark}
\usepackage{amsmath,amssymb,amsfonts,mathtools,amsthm,comment}
\usepackage{graphicx,tikz}
\usepackage{bbm}
\usepackage{physics}
\usepackage{tabularray}
\usepackage[capitalise]{cleveref}

\usepackage[
style=nature,
sorting=none,
maxbibnames=10
]{biblatex}
\AtEveryBibitem{
    \clearfield{urlyear}
    \clearfield{urlmonth}
    \clearfield{month}
    \clearfield{doi}
    \clearfield{url}
}

\let\cite\autocite

\usepackage{enumitem}
\usepackage[normalem]{ulem}

\DeclareMathOperator{\Wg}{Wg}
\DeclareMathOperator{\sym}{sym}
\DeclareMathOperator{\dia}{dia}
\DeclareMathOperator{\Comm}{Comm}
\DeclareMathOperator{\Or}{\mathcal{O}}
\DeclareMathOperator{\G}{G}

\DeclareMathOperator{\SbPD}{\widehat{\Pi}_\text{dist}}
\newcommand{\bs}{\boldsymbol}
\newcommand{\PD}{\Pi_{\text{dist}}}
\newcommand{\bPD}{\bs{\Pi}_{\text{dist}}}
\newcommand{\dens}[1]{\ensuremath{\ket{#1}\!\!\bra{#1}}}
\newcommand{\ct}{^\dagger}

\usepackage{authblk}
\newtheorem{theorem}{Theorem}
\newtheorem*{theorem*}{Theorem}
\newtheorem{corollary}[theorem]{Corollary}
\newtheorem{lemma}[theorem]{Lemma}

\newtheorem{definition}[theorem]{Definition}
\theoremstyle{definition}

\newcommand{\fro}{\mathrm{F}}

\DeclareMathOperator{\poly}{poly}
\DeclareMathOperator{\polylog}{poly\,log}

\DeclareMathOperator{\U}{U}
\renewcommand{\O}{\operatorname{O}}
\DeclareMathOperator{\SO}{SO}
\DeclareMathOperator{\USp}{USp}
\DeclareMathOperator{\Sp}{\USp} 
\DeclareMathOperator{\Cl}{Cl}
\DeclareMathOperator{\Pa}{P}
\DeclareMathOperator{\M}{M}

\newcommand{\EE}{\mathbb{E}}

\newcommand{\Z}{\mathbb{Z}}

\newcommand{\C}{\mathbb{C}}
\newcommand{\F}{\mathbb{F}}
\newcommand{\E}{\mathbb{E}}

\newcommand{\one}{\mathbbm{1}}

\newcommand{\mc}[1]{\mathcal{#1}}

\newcommand{\mcE}{\mc{E}}
\newcommand{\mcS}{\mc{S}}
\newcommand{\mcH}{\mc{H}}
\renewcommand{\H}{\mcH}

\newcommand{\myleft}{\mathopen{}\mathclose\bgroup\left}
\newcommand{\myright}{\aftergroup\egroup\right}

\DeclarePairedDelimiterX{\iiiNorm}[1]{\lvert}{\rvert}{%
  \delimsize\lvert\delimsize\lvert#1\delimsize\rvert\delimsize\rvert%
}

\DeclarePairedDelimiterXPP\snorm[1]{}\lVert\rVert{_\infty}{\ifblank{#1}{\,\cdot\,}{#1}}   

\DeclarePairedDelimiterXPP\twonorm[1]{}\lVert\rVert{_2}{\ifblank{#1}{\,\cdot\,}{#1}}   

\DeclarePairedDelimiterXPP\trnorm[1]{}\lVert\rVert{_1}{\ifblank{#1}{\,\cdot\,}{#1}}   

\DeclarePairedDelimiterXPP\fnorm[1]{}\lVert\rVert{_{\fro}}{\ifblank{#1}{\,\cdot\,}{#1}}   

\DeclarePairedDelimiterXPP\dnorm[1]{}\lVert\rVert{_\diamond}{\ifblank{#1}{\,\cdot\,}{#1}}   

\DeclarePairedDelimiterXPP\cbnorm[1]{}\lVert\rVert{_\mathrm{cb}}{\ifblank{#1}{\,\cdot\,}{#1}}   
\DeclarePairedDelimiterXPP\onenorm[1]{}\lVert\rVert{_{1\rightarrow 1}}{\ifblank{#1}{\,\cdot\,}{#1}}   
\DeclarePairedDelimiterXPP\ddnorm[1]{}\lVert\rVert{_{\diamond\rightarrow \diamond}}{\ifblank{#1}{\,\cdot\,}{#1}}   
\DeclarePairedDelimiterXPP\ssnorm[1]{}\lVert\rVert{_{\infty\rightarrow\infty}}{\ifblank{#1}{\,\cdot\,}{#1}}   

\DeclarePairedDelimiterX\Set[1]\{\}{%
  
  #1
}

\DeclarePairedDelimiterX\innerp[2]{\langle}{\rangle}{%
  \ifblank{#1}{\,\cdot\,}{#1} , \ifblank{#2}{\,\cdot\,}{#2}%
}

\DeclarePairedDelimiterX\sandwich[3]{\langle}{\rangle}%
  {#1\,\delimsize\vert\kern0.15ex\mathopen{}#2\kern0.15ex\delimsize\vert\kern0.15ex\mathopen{}#3}

\DeclarePairedDelimiterX\obraket[2]{(}{)}%
  {#1\kern0.15ex\delimsize\vert\kern0.15ex\mathopen{}#2}

\DeclarePairedDelimiterX\oketbra[2]{\vert}{\vert}%
  {#1\kern0.15ex\delimsize)\delimsize(\kern0.15ex\mathopen{}#2}

\DeclarePairedDelimiterX\osandwich[3]{(}{)}%
  {#1\,\delimsize\vert\kern0.15ex\mathopen{}#2\kern0.15ex\delimsize\vert\kern0.15ex\mathopen{}#3}

\newcommand{\tnd}[1]{^{\dagger, \otimes #1}}
\newcommand{\tnt}[1]{^{T, \otimes #1}}
\newcommand{\tn}[1]{^{\otimes #1}}

\begin{document}

\title{Will it glue? On short-depth designs beyond the unitary group}
\author[1]{Lorenzo Grevink}
\email{lorenzo.grevink@cwi.nl}
\author[2]{Jonas Haferkamp}
\author[3]{Markus Heinrich}
\author[1]{Jonas Helsen}
\author[4]{Marcel Hinsche}
\author[5]{Thomas Schuster}
\author[6]{Zoltán Zimborás}
\affil[1]{QuSoft and CWI, Amsterdam, The Netherlands}
\affil[2]{Saarland University, Saarbrücken, Germany}
\affil[3]{Institute for Theoretical Physics, University of Cologne, Cologne, Germany}
\affil[4]{Dahlem Center for Complex Quantum Systems, Freie Universität Berlin, Berlin, Germany}
\affil[5]{Walter Burke Institute for Theoretical Physics and Institute for Quantum Information and Matter,
California Institute of Technology, Pasadena, USA}
\affil[6]{QTF Centre of Excellence, Department of Physics,
University of Helsinki, P.O. Box 43, FI-00014 Helsinki, Finland}

\maketitle

\begin{abstract}
\noindent We study the formation of short-depth designs beyond the unitary group. 
    We provide a range of results on several groups of broad interest in quantum information science: the Clifford group, the orthogonal group, the unitary symplectic groups, and the matchgate group.
    For all of these groups, we prove that analogues of unitary designs cannot be generated by any circuit ensemble with light-cones that are smaller than the system size. %
    This implies linear lower bounds on the circuit depth in one-dimensional systems.
    For the {Clifford, orthogonal, and unitary symplectic group, we moreover show that commonly considered circuit ensembles cannot generate designs in sub-linear depth on any circuit architecture.}
    We show this by exploiting  observables in the higher-order commutants of each group, which allow one to distinguish any short-depth circuit from truly random.
    While these no-go results rule out short-depth designs over these subgroups, we prove that slightly weaker forms of randomness---including additive-error state designs and anti-concentration in sampling distributions---nevertheless emerge  at logarithmic depths in many cases. %
    Our results reveal that the onset of randomness in shallow quantum circuits is a widespread yet subtle phenomenon, dependent on the interplay between the group itself and the context of its application.
\end{abstract}

\section{Introduction}
Random quantum operations are a central tool in quantum information theory and beyond. Applications include a wide range of research fields such as 
quantum system characterization~\cite{
emerson_scalable_2005,helsen_general_2022,heinrich_randomized_2023}
,
quantum learning~\cite{huang_predicting_2020}
, 
quantum supremacy experiments~~\cite{arute2019quantum}, 
quantum machine learning~\cite{mcclean2018barren,larocca2025barren}
, 
quantum cryptography~\cite{kretschmer2021quantum,kretschmer2023quantum}
,
quantum many-body physics \cite{adam_operator_2018,fisher_random_2023}
,
and quantum gravity~\cite{hayden2007black,yoshida2017efficient,brown2018second}
. 

A natural, and the most commonly used, candidate for a random class of quantum operations are the Haar-random unitaries given by the unique uniform distribution on the unitary group.
While easy to work with in theory, the Haar measure is often far from practical:
Standard counting arguments show that exponential circuit depth is required to implement a Haar random unitary \cite{knill1995approximation}.
A second, independent objection to Haar-random unitaries is that they require access to a universal gate set, rendering them potentially expensive to implement on early fault-tolerant devices.
Significant effort has thus been spent to find low-depth constructions of random unitaries that appear Haar-random to a restricted observer~\cite{harrow_random_2009, brandao_local_2016, haferkamp_random_2022, chen2024efficient, chen_incompressibility_2024, schuster_random_2025,laracuente2024approximate,ma_how_2024,metger_simple_2024}.
Roughly, these constructions fall into two categories: Approximate unitary $k$-designs~\cite{gross2007evenly} that appear Haar-random to an observer with access to at most $k$ copies of the sampled unitary, and pseudorandom unitaries that appear Haar-random to a computationally bounded observer~\cite{ma_how_2024,ji_pseudorandom_2018}.

While approximate unitary designs and pseudorandom unitaries using circuits of depth linear in the number of qubits have been known for a few years \cite{brandao_local_2016}, sub-linear depth constructions were only found very recently \cite{schuster_random_2025,laracuente2024approximate}.
There, such ensembles are constructed in logarithmic and polylogarithmic depth, respectively, even in a $1\mathrm{D}$ geometry.
The results of Refs.~\cite{schuster_random_2025,laracuente2024approximate} came as a great surprise, because in this regime lightcones are of negligible size compared to the system size and yet the distinguishability to Haar-random unitaries (which generate near-maximally entangled states) decays exponentially.
Indeed, previous methods -- often based on analyzing the spectrum of so-called moment operators -- also apply to classical analogues of random quantum operations such as random reversible circuits.
But mixing in logarithmic depth appears to be a quantum phenomenon and can be proven to be false for classical analogues~\cite{schuster_random_2025}.
What is more, Ref.~\cite{schuster_random_2025} gives an argument that rules out log-depth construction of \textit{orthogonal} designs, i.e.~based on orthogonal matrices.
Not only do classical circuits not mix at logarithmic depth, but even the absence of complex numbers results in a linear lower bound.
Instead, the construction in Ref.~\cite{schuster_random_2025}, ``glues'' designs on local, log-sized patches together to form designs globally.
A dichotomy seems to emerge: Either, the class of operations is sufficiently universal to allow for a gluing construction or a linear lower bound can be proven.

In this paper we delineate the boundaries of this phenomenon and ask generally: \textit{Will it glue?} 
More precisely, we use lightcone arguments to prove strong lower bounds on the circuit depth for \emph{additive-error} group designs over the Clifford group, the orthogonal group, the unitary symplectic group, and the matchgate group. 
These results are particularly strong in one dimension, setting linear lower bounds on the required circuit depth.
In general architectures, these bounds depend on the connectivity in the form of the lightcone of qubits and still constitute non-trivial results.

Intriguingly, we are also able to show linear lower bounds \emph{independent} of the underlying architecture for a vast majority of circuit ensembles considered in the literature such as the popular ``brickwork'' circuits.
We derive these bounds for \emph{relative-error} state designs generated by such circuits which implies analogous circuit lower bounds for relative-error group designs.
These findings indicate that lightcones underestimate the mixing time of non-universal circuits, while Ref.~\cite{schuster_random_2025,laracuente2024approximate} show that they overestimate in the case of universal circuits.
In this sense, lightcone considerations may generally be misleading when studying designs.

However, for the strictly weaker notion of \emph{additive-error} state $k$-designs, we instead give logarithmic depth constructions for the orthogonal and symplectic, as well as the Clifford group (if $k<6$). Moreover, we show that for all of these cases, we have a generalized notion of anti-concentration in logarithmic depth. 
For the matchgate group, we show that linear depth is still necessary for additive-error state designs and we moreover give a (probably suboptimal) $\Omega(n^{1/3})$ depth lower bound for anti-concentration (relative to the corresponding value for the uniform matchgate group average), specifically for a ``gluing'' type construction.

\section{Summary of results}

In this paper, we study four subgroups of the $n$-qubit unitary group, namely the Clifford group $\Cl(n)$, the orthogonal group $\O(2^n)$, the unitary symplectic group $\USp(2^n)$, and the matchgate group $\M(n)$. 
For the four subgroups, we consider different notions of random ensembles, namely approximate group or state designs with additive or relative errors.
Here, all of these are assumed to be ensembles of unitaries from the respective subgroups. However, the local gates from which the unitaries are constructed need not be from the subgroup (cf.~Section \ref{sec:symplectic-group} in the Supplementary Information for examples in the symplectic case).
In addition, we investigate whether local circuits over the mentioned subgroups \emph{anti-concentrate} in logarithmic depth.
To be precise, Haar-random gates from $\Cl(n)$, $\O(2^n)$, and $\USp(2^n)$ do in fact anti-concentrate in the usual sense that $\EE_{U\sim G} |\sandwich{0}{U}{0}|^4 \leq c\, 4^{-n}$ for a constant $c\geq 1$, and thus it is expected that sufficiently deep random circuits do as well.
However, we adopt a slightly generalized notion of anti-concentration in this work: We say that a random circuit over a subgroup $G$ anti-concentrates up to order $k$ if the $G$-Haar value $\EE_{U\sim G} |\sandwich{0}{U}{0}|^{2k}$ is reproduced up to a multiplicative constant.\footnote{These higher moments are also referred to as ``inverse participation ratios'' in the many-body literature.}
This distinction is particularly relevant for the matchgate group which does not anti-concentrate in the conventional sense.

For all of these notions, we either show that they can be constructed in $\Or(\log(n))$ depth, or aim at deriving $\Omega(n)$ lower bounds, where $n$ denotes the number of qubits (we suppress the $k$ and $\epsilon$-dependence).
We obtain linear lower bounds unconditionally for 1D architectures, or for circuits constructed from local random gates (``generalized brickwork circuits'') in any architecture.

\begin{figure}
    \begin{minipage}[t]{0.55\linewidth}
        \vspace*{0mm}
        \includegraphics[width=\linewidth]{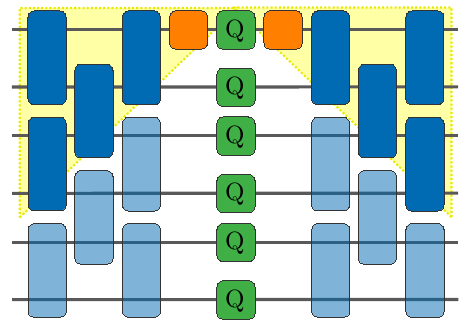}
    \end{minipage}%
    \begin{minipage}[t]{0.45\linewidth}
        \vspace*{0mm}
        \caption{Detection of short-depth local circuits over certain subgroups $G$ by a lightcone argument.
        We construct a factorizing operator $Q\tn{n}$ that commutes with the $k$-copy action of $G$.
        By `hiding' the first factor behind a single-qubit gate, a local circuit will no longer leave the perturbed operator invariant -- instead only the local gates outside the lightcone of the first qubit cancel due to the invariance of the local $Q$'s.
        If the circuit is too short and the lightcone does not contain all qubits, then some of the qubits stay unentangled with the rest. 
        This can be probed by a suitable observable, allowing to separate short-depth circuits from deep circuits and setting a lower bound on the depth of approximate designs.}
        \label{fig:lightcone}
    \end{minipage}
\end{figure}

Our lower bounds come in three different flavors, which all build on the existence of suitable invariances of the tensor power representation of $G$, i.e.~an operator $Q$ that commutes with $U\tn{k}$ for all $U\in G$.
Generalizing an observation in Ref.~\cite{schuster_random_2025}, we first use this to construct lightcone arguments that rule out short-depth circuits, cf.~Fig.~\ref{fig:lightcone}.
The bounds derived in this way scale with $\ell(n)$, i.e.\ the minimal depth of a lightcone that is supported on a constant fraction of all qubits.
They are thus particularly strong in 1D yielding $\Omega(n)$, but still give non-trivial lower bounds in other architectures, for instance $\Omega(n^{1/D})$ in $D$-dimensional lattices and $\Omega(\log n)$ in all-to-all connectivity.
Secondly, the same invariances can also be used to give lower bounds that depend on the \emph{diameter}, the maximal distance between any two qubit in the architecture, instead -- assuming that the ensemble is ``locally independent'', i.e.~gates on qubits that are far away from each other are drawn independently (cf.~Sec.~\ref{sec:methods} for an exact definition).
Finally, we go beyond architecture-dependent bounds by noting that the mentioned invariances allow to relate outcome probabilities to local Pauli expectation values and the latter may need linear depth to converge.
We show such a linear depth lower bound rigorously for the above mentioned generalized brickwork circuits.

Our upper bounds on the other hand all explicitly pass through the gluing lemma of Ref.~\cite{schuster_random_2025} (thus proving that it does ``glue''). We extend this construction by various non-trivial statements of \emph{unitary} $k$-design-ness of the Clifford, orthogonal and symplectic groups when restricting the input states to being positive partial transpose (PPT) on all copies of the state. 

Table~\ref{tab:results-overview} shows an overview of the results of this paper, together with a few  known results from the literature.
We note that we can construct unitary 1-designs in constant depth for every group, so results are only presented for $k\geq 2$.
We now discuss our results for each group individually.

\begin{table}
    \centering
    \begin{tblr}{
        columns={colsep=4pt},
        hline{1,3,Z}={1.2pt},
        hline{2,4-Y}={0.4pt},
        cells={c,m},
        cell{-}{1}={l,m},
        cell{3-6}{5}={r,m},
        cell{3-6}{6}={l,m},
        cell{1}{1}={r=2,c=1}{l,m},
        cell{1}{2}={c=5,r=1}{c,m},
        cell{2,Z}{5}={c=2,r=1}{c,m}
    }
    {Notion of\\ randomness} & Group\\
     & $\U(2^n)$ & $\O(2^n)$ &  $\USp(2^n)$ & $\Cl(n)$ && $\M(n)$ \\
    {additive-error\\ group designs}
        & $\Or(\log(\frac{n}{\epsilon}))$ \cite{schuster_random_2025,laracuente2024approximate}
        & $\Omega(\ell(n))$ 
        & $\Omega(\ell(n))$
        & {$k < 4$:\\ $k\geq 4$:} & {$\Or(\log(\frac{n}{\epsilon}))$  \cite{schuster_random_2025,laracuente2024approximate}\\ $\Omega(\ell(n))$}
        & $\Omega(n)$ \\
    {relative-error\\ state designs}
        & $\Theta(\log(\frac{n}{\epsilon}))$ \cite{schuster_random_2025,laracuente2024approximate}
        & {$\Omega(\dia)^\ast$\\$\Omega(n)^\dagger$}
        & $\Omega(\ell(n))$
        & {$k < 4$:\\$k\geq 4$:} & {$\Theta(\log(\frac{n}{\epsilon}))$  \cite{schuster_random_2025,laracuente2024approximate}\\$\Omega(\dia)^\ast,\Omega(n)^\dagger$}
        & $\Omega(n)$ \\
    {relative-error\\ group designs}
        & $\Theta(\log(\frac{n}{\epsilon}))$ \cite{schuster_random_2025,laracuente2024approximate}
        & all above
        & a.a., $\Omega(n)^\dagger$
        & {$k < 4$:\\$k\geq 4$:} & {$\Theta(\log(\frac{n}{\epsilon}))$  \cite{schuster_random_2025,laracuente2024approximate}\\ all above}
        & $\Omega(n)$ \\
    {additive-error\\ state designs}
        & $\Or(\log(\frac{n}{\epsilon}))$ \cite{schuster_random_2025,laracuente2024approximate}
        & $\Or(\log(\frac{n}{\epsilon}))$ 
        & $\Or(\log(\frac{n}{\epsilon}))$ 
        & $k < 6$: & $\Or(\log(\frac{n}{\epsilon}))$ 
        & $\Omega(n)$ \\
    {anti-\\concentration}
        & $\Theta(\log(n))$ \cite{barak_spoofing_2021,dalzell_random_2022,schuster_random_2025,laracuente2024approximate}
        & $\Or(\log(n))$ \cite{sauliere_universality_2025}
        & $\Or(\log(n))$ 
        & $\Theta(\log(n))$ \cite{magni_anticoncentration_2025,aharonov2023polynomial}
        && $\Omega(\frac{n^{1/3}}{\log(n)})$
    \end{tblr}
    \caption{Overview of circuit depth bounds obtained in this work and in the literature, for various notions of $k$-designs and different subgroups of the unitary group (bounds without references are derived in this work). 
    Here, we always assume $k\geq 2$ and indicate further restrictions on $k$ when necessary, but otherwise suppress the $k$-dependence  in the asymptotic notation.
    The first set of our lower bounds depend on the architecture through $\ell(n)$, the minimal depth of a lightcone that is supported on a constant fraction of all qubits.
    This translates to $\ell(n)=\Theta(n^{1/D})$ on a $D$-dimensional lattice and to $\ell(n)=\Theta(\log n)$ on all-to-all architectures.
    The other lower bounds are derived under mild structural assumption on the used circuit:
    The ones indicated by a $\ast$ depend instead on the diameter $\dia$, the largest distance of two qubits in the architecture, which can yield a different scaling as the lightcone bounds.
    The ones indicated by a $\dagger$ are derived for so-called ``generalized brickwork circuits'' and do no longer depend on the underlying architecture.
    The upper bounds are derived for the two-layer ``superblock'' circuits from Ref.~\cite{schuster_random_2025}.}\label{tab:results-overview}
\end{table}

\paragraph{Clifford group.}
We know that the multi-qubit Clifford group is an exact unitary $3$-design~\cite{zhu_clifford_2016}, thus the results from Refs.~\cite{schuster_random_2025,laracuente2024approximate} carry over to random Clifford circuits for $k\leq 3$. 
It is therefore interesting to see that a clear distinction can be made between Clifford $3$-designs and Clifford $4$-designs in terms of circuit depth.
Indeed, we prove that if the lightcone starting from the first qubit does not contain a constant fraction of all qubits, an ensemble cannot be an (additive and relative-error) Clifford 4-design (Theorem~\ref{theorem: clifford rel unitary} in the Supplementary Information). In the $1\mathrm{D}$ architecture this gives us a lower bound of $\Omega(n)$, which is exponentially worse than for $k=3$. For general architectures (e.g. all-to-all connectivity) this leads to a general lower bound of $\Omega(\log(n))$.

We also prove a lower bound of $\Omega(\dia)$ for relative-error Clifford state 4-designs, where $\dia$ is the architecture's diameter, under the additional assumption that the circuits are \emph{locally independent} (see Theorem~\ref{theorem:clifford_rel_state} in the Supplementary Information).
This is typically the case for the commonly considered random quantum circuits and thus a minor restriction only.
However, for Clifford \emph{group} designs, this assumption is not needed.
We expect that this can be generalized to Clifford state designs and thus the assumption of local independence can be eventually dropped.

Finally, we also prove a strong $\Omega(n)$ lower bound on relative-error state designs generated by generalized brickwork circuits \emph{in any architecture}. Of course, this implies the same lower bound for relative-error group designs for this class of random circuits. 
Again, we expect that the assumptions on the circuit ensemble can be considerably relaxed.

Given these linear lower bounds, it is somewhat surprising that the construction of additive-error Clifford state $4$- and $5$-designs of depth $\Or(\log(n/\epsilon))$ in $1\mathrm{D}$ is possible (Corollary~\ref{cor:clifford add state} in the Supplementary Information).
This result follows from the fact that, for $k\leq 5$, the $k$-copy ensemble of stabilizer states is close in trace distance to Haar-random states with an error $\Or(2^{-n})$ (Lemma~\ref{lem:clifford_state_ppt} in the Supplementary Information).
It is notable that our proof technique fails for $k\geq 6$, as there stabilizer states can be efficiently distinguished from Haar random states through Bell difference sampling. 
After the release of the first version of this work, the existence of $\Or(\log(n))$-depth additive-error Clifford state designs for $k\geq 6$ was proven in Ref.~\cite{zhang_designs_2025}

\paragraph{Orthogonal group.}
A lower bound on the depth of relative-error orthogonal group designs was already known from Ref.~\cite{schuster_random_2025}, based on a similar lightcone argument as we use in this work.
We extend the lower bound $\Omega(\ell(n))$ to an additive-error statement in Theorem~\ref{theorem: orthogonal rel group design} in the Supplementary Information.

We also extend their argument to relative-error orthogonal state designs in Theorem~\ref{theorem:orthogonal_rel_state} in the Supplementary Information, provided the same nuances on local independence as in the Clifford case. We moreover also prove the same strong $\Omega(n)$ depth lower bound for generalized random brickwork circuits in any architecture.

Surprising is the fact that we can construct additive-error state designs in $\Or(\log(n/\epsilon))$ depth, see Corollary~\ref{cor:orth add state} in the Supplementary Information. The theorem states that the unitary and orthogonal Haar twirl are equal on PPT states, up to an additive error exponentially small in $n$, more precisely, with $n$-dependence $\frac{1}{2^{n/2}}$. Note that an additive error of $\epsilon$ implies a relative error of $\epsilon \cdot 4^{nk}$~\cite{brandao2016local}. The exponent of our additive error is thus not small enough to imply a relative error statement.

We also show that random orthogonal circuits anti-concentrate in depth $\Or(\log(n))$ for $k=2$ at the example of the two-layer ``superblock'' circuits from Ref.~\cite{schuster_random_2025}. 
We note that anti-concentration at higher $k$ was shown for orthogonal MPS in Ref.~\cite{sauliere_universality_2025} and we think that the analogous statement for superblock circuits can also be derived along the lines of Ref.~\cite{magni_anticoncentration_2025}.

\paragraph{Symplectic group.}
We prove that the lower bound on the depth of additive-error orthogonal group designs also holds for symplectic group designs. For relative-error state designs we prove a result slightly stronger than for the orthogonal case, i.e.\ we do not need the assumption of locally independence. Also, the result does not depend on the diameter, but on the size of the lightcones. 
We finally prove a linear lower bound on relative-error symplectic \emph{group} designs for generalized random brickwork circuits in any architecture.
Note that we obtained a slightly stronger statement, namely for relative-error state designs, for the orthogonal and Clifford group.

The result on state designs is surprising, as we know that Haar-random ensembles of symplectic and unitary states are indistinguishable~\cite{west_random_2024}, i.e.\ symplectic state designs are the same as unitary state designs. We know that relative-error state designs can be constructed in $\Or(\log(n))$ depth using local unitary gates in the $1\mathrm{D}$ architecture~\cite{schuster_random_2025}. It is therefore an interesting result that relative-error state designs cannot be constructed in $\Or(\log(n))$ depth using local symplectic gates. The indistinguishability between symplectic and unitary state designs combined with the unitary gluing lemma~\cite{schuster_random_2025} does however imply symplectic anti-concentration in $\Or(\log(n))$ depth.

We also prove that additive-error symplectic state designs can be constructed in $\Or(\log(n/\epsilon))$ depth, in a similar way as in the orthogonal case.

\paragraph{Matchgate group.}
Finally we provide lower bounds for the matchgate group. We prove two main results. The first result, given in Theorem~\ref{theorem:matchgate_anti-concentrate} in the Supplementary Information states that the superblock architecture does not attain the minimal anti-concentration value (which is $\Theta(\sqrt{n})$) for any circuit depth that is $o(n^{1/3}/\log(n))$, and thus anti-concentration requires $\Omega(n^{1/3}/\log(n))$ circuit depth. We suspect that this scaling behavior is non-optimal and that a linear depth lower bound should be possible. Interestingly our proof is entirely combinatorial in nature and uses no lightcone arguments. The second result, given in Theorem~\ref{thm: matchgate state designs} in the Supplementary Information states that matchgate relative error state designs are not possible for any depth less than $\Omega(n)$, with the minor technical caveat that the proof only covers errors (in both the relative and additive cases) of size $O(1/n)$. This argument is essentially lightcone based, but is substantially less subtle than similar arguments for other groups as the matchgate group is inherently local with respect to an $1\mathrm{D}$ architecture. 

\paragraph{Outlook.}
Our results immediately impact constructions of efficient randomization protocols over restricted groups, for instance:
\begin{itemize}
    \item The circuit complexity of randomized protocols based on Clifford unitaries under circuit re-use (such as thrifty shadow estimation \cite{helsen_thrifty_2023}).
    \item Our results prevent the existence of short-depth matchgate designs, thus fermionic protocols, such as fermionic shadows, may need much deeper circuits than their qubit counterparts. This answer a question raised in Ref.~\cite{chapman2025fermionic} in the negative (at least for geometrically local circuits).
\end{itemize}
Beyond this, our work raises a number of interesting theoretical questions of which we list several key ones in the following:
\begin{itemize}
    \item Our work raises the question which subgroups of the unitary group could actually support a "gluing lemma". Beyond its strict meaning, gluing-type constructions succeed in many settings as demonstrated in this work. We believe that at least the strong existence question could be rigorously decided, as there are rather few "relevant" groups. Perhaps this could be done following a logic similar to the work of Tiep and Guralnick \cite{guralnick2005decompositions}. However, we suspect that any of these questions are highly nontrivial to answer. 
    \item Beyond the ones considered here, there are natural families of structured circuit distributions that uniformly approximate the Haar measure on the unitary group. Chief among these are so called T-doped Clifford circuits \cite{haferkamp_efficient_2023}. The key open question here is to characterize the density of T gates required to transition from a "lightcone phase" (c.f. the Clifford group) to a "gluing phase" (c.f. the unitary group). We use the language of phase transition here, but it is not clear if this will be a phase transition in the conventional sense.
    \item There is, for a fixed group, also the question which classes of observables (acting on several copies of the unitary) actually witness lightcones. We see that anti-concentration is generally possible in short depth even when designs cannot be formed. Is there a natural class of anti-concentration-like observables? We suspect that the PPT condition we use in several proofs could serve as at least a sufficient condition, but a general statement of that form is unclear.  
\end{itemize}

\paragraph{Related work.}
While finishing this manuscript, we became aware of the work~\cite{west_no-go_2025}, which shares several goals and results with our work.
There, it is shown that $G$-invariant pure states (in the sense that $U^{\otimes k}\ket{\psi} \propto \ket\psi$) can be used to discriminate short-depth circuits from $G$-Haar-random unitaries, which yields a lower bound on the diamond distance between the respective twirls. 
This implies that there cannot be additive-error designs over $G$ in such depth.
Both our work and Ref.~\cite{west_no-go_2025} consider the same set of examples for $G$ and arguments based on lightcones.
However, we consider a more general notion of invariance, namely the existence of factorizing positive-semidefinite operators $Q$ in the $k$-fold commutant (i.e.~$[U^{\otimes k},Q]=0$).
This allows us to already give lower bounds on Clifford 4-designs, while the criterion in Ref.~\cite{west_no-go_2025} only applies to 8-designs.
Moreover, we also rule out relative-error state designs over the considered subgroups, which is not implied by Ref.~\cite{west_no-go_2025}. 
Going beyond lightcone arguments, we prove stronger results for generalized random brickwork circuits which set linear lower bounds on relative designs in any architecture.
Finally, we also construct additive-error state designs over all but the matchgate groups in logarithmic depth, and show that the orthogonal, unitary symplectic, and matchgate group anti-concentrate in logarithmic depth (in the generalized sense).

\section{Methods}
\label{sec:methods}

In the following, we will give an overview of the methods used in this work and sketch the arguments used to derive our results. 
We refer the reader to the Supplementary Information for the exact statements of our results and detailed proofs.
We will begin by introducing some notation and the various notions of randomness we use in our results.

\subsection{Notation and concepts}\label{section: preliminaries}

Let $G$ be a (compact) subgroup of the unitary group $\U(2^n)$ on the $n$-qubit Hilbert space $\H = (\C^2)^{\otimes n}$ (in this work, we focus on $n$-qubit systems, although some of our results can be readily generalized to qudit systems).
The Haar-measure on $G$ should be thought of as the ``most uniform'' probability measure on $G$ -- Formally, it is the unique probability measure that is both left and right invariant under multiplication with elements of $G$.
We write $\E_{U \sim G}$ for the expectation value over the Haar measure on $G$. 

We often call general probability distributions over $G$ \emph{ensembles} and typically denote them by $\mcE$.
These will always be implemented by quantum circuits, so we will also sometimes speak of circuit ensembles. 
To quantify how close an ensemble $\mcE$ is to the Haar measure, we consider \emph{moments} of the distribution, given as
\begin{equation}
\label{eq:moment}
  \E_{U \sim \mathcal{E}} \tr ( B U\tn{k} A U\tnd{k} ) \,,
\end{equation}
for suitably sized matrices $A$ and $B$. 
If all $k$-th moments of $\mcE$ agree with those of the Haar measure, we call $\mcE$ a \emph{unitary $k$-design over the group $G$}.
We note that if $A$ (or $B$) commute with $U^{\otimes k}$ for all $U\in G$ (we say they are in the $k$-th commutant of $G$), then the moments trivially agree.

We typically only require that moments are reproduced up to a (small) error -- we then call $\mcE$ an \emph{approximate $k$-design}.
However, there are different notions of approximation, which are not equivalent in their resource demands.
\emph{Additive-error unitary $k$-designs over $G$} require that $k$-th moments are reproduced up to an additive error.
For this to be well-defined, the matrices $A$ and $B$ in Eq.~\eqref{eq:moment} have to be suitably restricted.
Here, we typically take $A$ to be a quantum state and $B$ to be a POVM element -- such that the error relates to the probability of distinguishing the averaging (or twirling) channels $\Phi_{\mcE}^{(k)} := \EE_{U\sim\mcE} U\tn{k} (\cdot) U\tnd{k}$ and $\Phi_{G}^{(k)} := \EE_{U\sim G} U\tn{k} (\cdot) U\tnd{k}$ and can be alternatively written in terms of the \emph{diamond norm} as $\dnorm{\Phi_{G}^{(k)} - \Phi_{\mathcal{E}}^{(k)}}$.\footnote{Strictly speaking, we would need to include auxillary systems in Eq.~\eqref{eq:moment}, but we omit this detail for this discussion.}
Similarly, \emph{relative-error unitary $k$-designs over $G$} require the reproduction of $k$-th moments up to a relative or multiplicative error.
We keep $A$ and $B$ to be restricted, but as overall normalizations do not matter for relative errors, we take them to be arbitary positive semi-definite (psd) matrices (i.e.~Hermitian matrices with non-negative eigenvalues).
A relative error $\epsilon$ always implies an additive error $2\epsilon$ \cite{brandao_local_2016}.
Even more, given a relative-error $\epsilon$-approximate unitary $k$-design $\mathcal{E}$, any quantum experiment that queries a unitary $U$ from $\mathcal{E}$ at most $k$ times outputs a state that is, on average, $2\epsilon$-close in trace distance to the output state obtained if $U$ were sampled from the Haar measure instead.
For an additive-error $\epsilon$-approximate unitary $k$-design, the same statement holds but then for non-adaptive quantum experiments~\cite{schuster_random_2025}. An additive error is thus a weaker requirement than a relative error.

We have analogous statements for state ensembles:
Haar-random quantum states over $G$ are given by $\mcS_G = \{ U\ket{0^n} \, | \, U \in G \}$ with the induced measure.
The moments of a state ensemble $\mcE'$ are given as 
\begin{equation}
   \E_{\psi \sim \mcE'} \tr ( B \ketbra{\psi}\tn{k} ) \,.
\end{equation}
\emph{Additive-error state $k$-designs over $G$} are those for which these moments are $\epsilon$-close to the Haar moments (for any POVM element $B$), or equivalently the one with trace distance error $\trnorm{\E_{\psi \sim \mcS_G} \ketbra{\psi}\tn{k} - \E_{\psi \sim \mcE'} \ketbra{\psi}\tn{k}} \leq \epsilon$.
Finally, \emph{relative-error state $k$-designs over $G$} reproduce all $k$-th moments with psd $B$ up to relative error.
Like with group designs, every quantum experiment that queries a state $\ket\psi$ at most $k$ times from a relative-error state $k$-design $\mathcal{E}'$, outputs a state that is, on average, $2\epsilon$-close to the state we would have gotten if we would have sampled from the Haar random distribution. For additive-error a similar statement holds, but then for non-adaptive algorithms. 
If $\mathcal{E}$ is a relative-error unitary $k$-design, then $\{U\ket{0^n}\mid U \sim \mathcal{E}\}$ is a relative-error state $k$-design. A similar statement holds for additive error. State designs are thus a weaker notion than group designs.

A weaker notion of randomness is \emph{anti-concentration}. 
As mentioned above, we here adopt a common misusage of this word in the random circuit community and mean that the moments \eqref{eq:moment} of the ensemble $\mathcal E$ with $A=B=\ketbra{0^n}\tn{k}$ approximate the Haar value over $G$ up to a relative error.
We then say that $\mcE$ anti-concentrates up to order $k$.
In contrast, anti-concentration in the conventional sense is a property of the ensemble of outcome distributions.
The conventional notion  is implied if the $k=2$ Haar moment scales as $4^{-n}$ which is however not necessarily true for any subgroup $G$.
For the here studied subgroups, this is the case for the matchgate group.
We note that if $\mathcal{E}$ forms a relative-error state $k$-design, then it anti-concentrates up to order $k$.

Most of our bounds depend on the \textit{architecture} according to which the qubits are ordered and connected.
An architecture is described by a graph, where the vertices of the graph are the qubits, and an edge between two vertices indicates that a 2-qubit gate can be applied on those 2 qubits. Commonly studied architectures are $1\mathrm{D}$ and $2\mathrm{D}$ lattices (with open or closed boundary conditions), or all-to-all connectivity (the complete graph). 
We define the \textit{diameter} (often denoted as $\dia$) of an architecture as the largest possible distance between two vertices, e.g.\ the diameter of $1\mathrm{D}$ is $\dia = \Theta(n)$, for $2\mathrm{D}$ $\dia = \Theta(\sqrt{n})$, and for the all-to-all connectivity $\dia = \Theta(1)$.
Note, however, that many known constructions of designs in the all-to-all connectivity actually use architectures of higher diameter. For example, the construction of Clifford 3-designs in Ref.~\cite{schuster_random_2025} in $\Or(\log\log(n))$ depth implicitly uses an architecture of diameter $\Or(n/\log(n))$.

Finally, we elaborate a bit on the assumptions we impose on the circuit ensembles to derive lower bounds that are independent of the architecture.
Due to their clearly defined structure, we obtain the strongest results for so-called \textit{generalized random brickwork circuits}.
We define them as a circuit ensemble that is constructed from independent $r$-local Haar-random gates, for some constant $r$.

\begin{definition}[Generalized random brickworks]\label{def:brickwork}
    Let $A_1,\ldots, A_L$ be partitions of $[n]$, where each $a_{i,j}\in A_{i}$ is a subset of $[n]$ of size at most some constant $r$, such that the qubits in $a_{i,j}$ are connected w.r.t to the underlying architecture. For each $a_{i,j}$ draw a unitary $U_{i,j}$ uniformly at random from the group $G$ on the qubits $a_{i,j}$ or apply the identity . Let $\mathcal{E}$ denote the ensemble of circuits generated by composing these unitaries in $L$ layers. We call an ensemble of this form a \textbf{fixed brickwork}. A \textbf{generalized random brickwork} is defined as a convex combination of (different) fixed brickworks.
\end{definition}

We emphasize that architecture-independent lower bounds derived for generalized random brickworks still hold under a significant relaxation of this definition, and we decided for Def.~\ref{def:brickwork} to simplify the presentation.
More precisely, we do not necessarily need that the local gates are uniformly distributed and their distributions do not have to be identical either (but they should be sampled independently).
Instead, our proofs require that the local gates preserve some single-qubit Pauli operator (up to a phase) with a finite probability or, more generally, that they can be replaced by such a distribution up to second moments (see also Sec.~\ref{sec:methods lower bounds} for a sketch of the argument).
This includes a vast majority of the random circuits considered in the literature.

See~Refs.~\cite{harrow_approximate_2023, dalzell_random_2022, haferkamp_random_2022, chen_incompressibility_2024}, and Corollary 1, first bullet point, in Ref.~\cite{schuster_random_2025} for examples of fixed brickworks, that are thus generalized random brickworks as well.
See Definition 1.4, first bullet point, in Ref.~\cite{chen_incompressibility_2024} for an example of a generalized random brickwork which is not a fixed brickwork.
See Ref.~\cite{suzuki_more_2024} (with CNOT as entangling gate) for an example that is a generalized random brickwork circuit under the relaxed definition .
See Refs.~\cite{metger_simple_2024, jiang_optimal_2020} and Corollary 1, second bullet point, in Ref.~\cite{schuster_random_2025} for examples of ensembles that are not generalized random brickworks.

We define a \textit{locally independent ensemble} as an ensemble such that the local gates acting on a qubit are drawn independently from the local gates acting on qubits that are ``far away''.
\begin{definition}[Locally independent ensemble]\label{def:local}
    Let $\mathcal{E}$ be an ensemble of unitary circuits on a fixed architecture.
    Let $S_x(U) \coloneq \{U^x_1, \ldots , U^x_{m_{U}}\}_{U}$ be the set of local unitaries that act on qubit $x$ for a circuit $U$. If for every pair of qubits $x,y$ that have distance at least $\dia/2$ we have that $S_x(U)$ and $S_y(U)$ are independent variables when $U$ is sampled from $\mathcal{E}$, then $\mathcal{E}$ is \textbf{locally independent}.
\end{definition}
Note that the notion of a locally independent ensemble does not make any sense in the all-to-all architecture, as there we have that $\dia = 1$.
Almost all of the ensembles named in the previous paragraph are locally independent. 
Only the ones in Ref.~\cite{ghosh_random_2024} and Definition 1.4, first bullet point, in Ref.~\cite{chen_incompressibility_2024} are not, as these ensembles use the all-to-all architecture.

\subsection{Lower bounds}
\label{sec:methods lower bounds}
Most of the lower bounds in this paper are based on a similar lightcone technique as used in Ref.~\cite{schuster_random_2025} to give a lower bound on the depth of orthogonal designs.
The exception are the lower bounds for generalized random brickwork circuits, on which we will elaborate more below. 

Recall from above that for $\mathcal{E}$ to be a relative-error group design over $G$, we require that expressions of the form $\tr(P U\tn{k}Q U\tnd{k})$, where $P,Q\geq 0$ are positive-semidefinite (psd) operators, with averages over $U$ either taken from $\mathcal{E}$ or from $G$, are close in relative error.
However, we will argue in the following that this generally requires linear-depth circuits to be accomplished, by relying on the existence of specific psd operators in the group's $k$-th order commutant (which are not in the unitary commutant).

More specifically, for the here considered $G$, we exploit that we can even find \emph{factorizing} $Q$ in the sense that $Q=Q_1\otimes\dots\otimes Q_n$, where all tensor products of the local operators $Q_i\geq 0$ are in the $k$-th commutant of $G$ restricted to the respective qubits.
Let us now add a small perturbation to $Q$, for instance by applying a Pauli-$Z$ matrix to the first qubit (on the first copy) to get $Q' \coloneqq (Z_1 \otimes \one^{\otimes(k-1)})Q(Z_1 \otimes \one^{\otimes(k-1)})$. 
Now $Q'$ is still psd but, in general, not invariant under $G$. 
However, everywhere except on the first qubit it is, i.e.\ if $V$ is a gate in $G$ that acts trivially on the first qubit, then $V Q' V^\dagger = Q'$. 

Hence, if we consider the action of a local $G$-circuit $U$ on $Q'$, all gates which are outside of the lightcone of the first qubit cancel:
This is most clearly seen at the example of a brickwork circuit, cf.~Fig.~\ref{fig:lightcone}. 
In the first layer of $U$, all gates cancel after conjugation of $Q'$ except the gates that act on the first qubit.
In the second layer, again all gates cancel except the ones that acts on qubits which are acted upon by the first layer, i.e.~the first two.
Iterating this argument shows that the only gates left after $d$ layers are the gates that form a lightcone starting from the first qubit.
If this lightcone does not include (a constant fraction of) all qubits, we are left with the operator $Q$ on the remaining qubits, unentangled with the remaining systems, which we can detect by measuring a suitable observable.
More precisely, we show that this gives means to distinguish deep circuits (or Haar-random unitaries) from short circuits, setting a lower bound of $\Omega(l(n))$ for the depth of $k$-designs for the specific group, where $l(n)$ is the minimal depth of a lightcone that is supported on a constant fraction of all qubits. In the $1\mathrm{D}$ architecture this implies a lower bound of $\Omega(n)$, in the all-to-all connectivity a lower bound of $\Omega(\log(n))$.

We note that finding such $Q$ for the considered groups is not difficult. 
For the orthogonal and unitary symplectic group, we can simply vectorize the preserved bilinear forms (the identify matrix $\one_n$ or the symplectic form $J_n$) to construct $Q$:
We obtain $Q$ in the $k=2$ commutant as the projector onto the maximally entangled state $\ket{\Omega_n}$ or on $\ket{J_n} \equiv (J_n\otimes \one_n) \ket{\Omega_n}$, respectively, which both factorize in the desired way.
For the Clifford group, we can find a suitable $Q$ in the $k=4$ commutant as the infamous stabilizer code projector $\propto\sum_P P^{\otimes 4}$, where the sum runs over all $n$-qubit Pauli operators.

The second lower-bound technique is built on the insight that the existence of factorizing $Q$ in the commutant allows to reduce some of the `psd' quantities $\tr(P U\tn{k}Q' U\tnd{k})$ to \emph{Pauli expectation values}.
The former are exactly the functions which can be estimated up to relative error using log-depth random circuits over the unitary group \cite{schuster_random_2025,laracuente2024approximate}, while the latter generally need linear depth, even for universal circuits.
On a high level, the existence of such a relation between these two types of functions for the here considered subgroups prevents the construction of log-depth (relative-error) designs.
More precisely, let us choose $P=P_1\otimes\cdots\otimes P_k$ to factorize over the $k$ copies such that
\begin{align}
 \tr(P U\tn{k}Q' U\tnd{k})
 &=
 \tr(P U\tn{k} (Z_1\otimes\one^{\otimes(k-1)}) U\tnd{k} Q U\tn{k} (Z_1\otimes\one^{\otimes(k-1)}) U\tnd{k})\notag \\
 &=
 \tr(P_1 (UZ_1 U^\dagger) \tilde Q (UZ_1 U^\dagger)) \,,
\end{align}
where $\tilde Q = \tr_{2,\dots,k}(\one\otimes P_2\otimes\cdots\otimes P_k Q)$ is the reduced operator on the first system.
Using the factorizing structure of $Q$, it is simple to find $P_i$ such that $\tilde Q \propto \ketbra{0}{0}$ and the above moment simply becomes $\propto\sandwich{0}{UZ_1 U^\dagger}{0}|^4$.
For the Clifford and orthogonal group, we can take $P=\ketbra{0^n}{0^n}\tn{k}$ (with $k=4$ and $k=2$, respectively) while for the symplectic group, we take $P=\ketbra{0^n}{0^n}\otimes\ketbra{1^n}{1^n}$.

For the previously introduced generalized random brickwork circuits, we then explicitly show that these Pauli expectation values need linear depth to approximate their Haar values within relative error.
Consequently, we obtain a lower bound of $\Omega(n)$ on the depth of relative-error state designs for the Clifford and orthogonal group, and relative-error group designs for the symplectic group.
Crucially, this result holds for generalized random brickwork circuits in \textit{any architecture} and is thus considerably stronger than the previously discussed lightcone arguments.
We are confident that our arguments can be generalized to an even broader class of circuit ensembles and may even hold for general circuits (cf.~the comment below Def.~\ref{def:brickwork}).

To prove these lower bounds, we exploit that we deal with second moments only and that we can thus replace the average action of each local gate with an (exact) 2-design over the respective group.
For all considered groups, we can take these 2-designs to be composed of Cliffords.
Then, there is a finite probability in every layer that the locality of $Z_1$ is preserved and thus the probability that the final operator is $Z_1$ decays exponentially with the depth.
To achieve the exponentially small Haar value, the depth thus has to be linear.\\

Finally, we use slightly different techniques to deal with the case of the matchgate group. The inherently one-dimensional nature of matchgates makes them relatively straightforward to deal with. We prove a lower bound on the anti-concentration value by relating it to a path counting problem in the commutant of the second moment averaged circuit. We cut a short depth circuit into roughly square pieces and lower bound the number of paths that pass through such cuts, which leads directly to a lower bound on the anticoncentration value. Our proof only deals with superblock circuits, but we believe the technique can be easily generalized to arbitrary matchgate circuit distributions.  Finally the absence of state designs is proven by a straightforward one-dimensional lightcone argument at the level of Majorana operators.    

\subsection{Upper bounds}\label{section: upper bounds}
The constructions for $\Or(\log(n))$-depth designs (of various types) in this paper are similar to the construction in~\cite{schuster_random_2025}.
The idea is to construct a two-layer brickwork circuit in which the local gates act on a large number of qubits (``superblocks'').
We organize $n$ qubits in $m$ local patches, each consisting of $\xi = n/m$ qubits. We choose $m$ to be even. In the first layer we apply superblocks to the local patches $(1,2), (3,4),\ldots, (m-1, m)$, and in the second layer to the local patches $(2,3), (4,5),\ldots, (m, 1)$.
In the first layer we thus have $m/2$ superblocks, that act on $2\xi$ qubits each. The second layer is the same as the first layer, except that the superblocks act on patches that are shifted $\xi$ qubits from the patches the superblock in the first layer act on. 
The superblocks are independently drawn from the Haar measure of a given group $G$ (or from $k$-designs if required).
The key step in Ref.~\cite{schuster_random_2025} is then the \textit{gluing lemma}.

\begin{lemma}[Gluing lemma, Theorem~$1$ in Ref.~\cite{schuster_random_2025}]
\label{lem:gluing}
    Let $\mathcal{E}$ be the ensemble of matrices drawn from the two-layer superblock architecture for the unitary group $\U$. Choosing $\xi = \Or(\log(nk/\epsilon))$, then $\mathcal{E}$ is a relative-error $\epsilon$-approximate unitary $k$-design.
\end{lemma}

In Ref.~\cite{chen_incompressibility_2024} it is shown that random brickwork circuits form unitary $k$-designs in depth~$\Or((nk + \log(1/\epsilon))\polylog(k))$.
If we thus insert random brickwork circuits of this depth in the superblock construction $\mathcal{E}$, we get a relative-error $\epsilon$-approximate unitary $k$-design in depth $d = \Or(\log(n/\epsilon) k\polylog(k))$. Note that this construction is a generalized random brickwork.

To prove that additive-error state designs can be constructed in depth $\Or(\log(n))$ for the orthogonal and the symplectic group,  we use a similar construction in conjection with a new technique:
We show that the Haar twirl over these groups is $\epsilon$-close in trace distance to the unitary Haar twirl, for $\epsilon$ polynomial in $k$ and $2^{-n}$, provided that the input state obeys the positive partial transpose (PPT) condition on every copy of the Hilbert space~$\H$.
This is proven by projecting on a subspace of the Hilbert space where the commutant of the considered subgroup is identical to the commutant of the unitary group. For the orthogonal group this is the distinct subspace~\cite{metger_simple_2024} and we define an analogus subspace for the symplectic group. 
We then observe that if $\rho$ is PPT, then $U\tn{k}\rho U\tn{k,\dagger}$ is as well. The two-layer superblock brickwork will thus provide a $\Or(\log(n))$-depth construction for additive-error state designs by the gluing lemma.

To obtain expressive circuit depth upper bounds on anti-concentration, we have to restrict the considered circuits to a relevant class of interest.
This is because we can otherwise find constant-depth circuits that anti-concentrate but are otherwise not interesting (for instance, a layer of parallel Hadamard gates  trivially anti-concentrates for the unitary group).
For concreteness, we here focus on the two-layer superblock construction from Ref.~\cite{schuster_random_2025}.
The bounds then follow by an explicit calculation using Weingarten calculus in the orthogonal case.
In the unitary symplectic case, this follows readily from the fact that the twirl over $\USp$-Haar random states is the same as over $\U$-Haar random states and thus the result follows directly from the gluing lemma.

\section*{Acknowledgements}
We are tremendously grateful to Yuzhen Zhang, Yimu Bao, and Yingfei Gu for several clarifying discussions regarding short-depth Clifford state designs for $k > 5$, and to Xhek Turkeshi and Ingo Roth for general discussions on short-depth random circuits. We are also grateful to the authors of Ref.~\cite{west_no-go_2025} for the motivation to finish our manuscript with greater than usual alacrity.
L.~Grevink and J.~Helsen acknowledge funding from NWO through NGF.1623.23.005 and Veni.222.331.
T.~Schuster acknowledges support from the Walter Burke Institute for Theoretical Physics at Caltech and the U.S. Department of Energy, Office of Science, National Quantum Information Science Research Centers, Quantum Systems Accelerator.
The Institute for Quantum Information and Matter, with which T.~Schuster is affiliated, is an NSF Physics Frontiers Center (NSF Grant PHY-2317110).
M.~Hinsche acknowledges funding by the German Federal Ministry for Education and Research (MUNIQC-Atoms).
M.~Heinrich acknowledges funding by the Deutsche Forschungsgemeinschaft (DFG, German Research Foundation) -- 547595784 and under Germany’s Excellence Strategy – Cluster of Excellence Matter and Light for Quantum Computing (ML4Q) EXC 2004/1 -- 390534769. 
Z. Zimbor\'as was supported by the Horizon Europe programme HORIZON-CL4-2022-QUANTUM-01-SGA via
the project 101113946 OpenSuperQPlus100 and the
QuantERA II project HQCC-101017733 as well as by the NKFIH OTKA grant FK 135220.
A substantial part of this work was developed during the Random Quantum Circuits workshop 2024 in Amsterdam.


\clearpage

\addcontentsline{toc}{section}{Supplementary Information}

 \begin{center}
 \textbf{\Large\sffamily Will it glue? On short-depth designs beyond the unitary group}\\[1em]
 \textit{\large -- Supplementary Information --}
 \end{center}
 \vspace{2em}

\setcounter{equation}{0}
\setcounter{figure}{0}
\setcounter{table}{0}
\setcounter{section}{0}
\setcounter{theorem}{0}
\makeatletter
\renewcommand{\theequation}{S\arabic{equation}}
\renewcommand{\thefigure}{S\arabic{figure}}
\renewcommand{\thesection}{S\arabic{section}}
\renewcommand{\thetheorem}{S\arabic{theorem}}

\tableofcontents

\section{Notation and definitions}\label{section: preliminaries}
Let $\H = \C^{D}$ be a $D$-dimensional Hilbert space. In this paper we focus on the qubit case, in which $D = 2^n$, where $n$ denotes the number of qubits. 
Let $X, Y, Z$ denote the Pauli matrices and let $ Z_i $ denote the operator $ \one^{\otimes (i-1)} \otimes Z \otimes \one^{\otimes (n-i)} $, i.e., the Pauli-$Z$ acting on the $i$-th qubit and identity elsewhere on $n$ qubits and analogously for $X_i,Y_i$.

\paragraph{Subgroups of the unitary group and Haar measure}
Let $G(D)$ be a subgroup of the unitary group $\U(D)$ on $\H$. The Haar-measure on $G(D)$ is defined as the unique probability measure that is both left and right invariant under multiplication with elements of $G(D)$, i.e.\ that for all $V \in G(D)$ and for every integrable function $f$ we have that
\begin{equation}
    \int_{G(D)} f(U) d\mu_H (D) = \int_{G(D)} f(VU) d\mu_H (D) = \int_{G(D)} f(UV) d\mu_H (D).
\end{equation}
We will often write $\E_{U \sim G}(f(U)) \coloneqq \int_{G(D)} f(U) d\mu_H (D)$ for integration over the Haar measure. Likewise, we can define the Haar measure on the state space $W_G \coloneqq \{\ket\psi \mid \exists g \in G; \ket\psi = g\ket{0_n}\}$, as the unique probability measure that is invariant under multiplication with elements of $G$.
By a slight abuse of notation we will write $\ket\psi\sim G$ for sampling from the Haar measure on $W_G$.

\paragraph{Ensembles and circuit depth}
We call distributions of elements of a group \emph{ensembles} (typically denoted by $\mathcal{E}$). These group elements will always be implemented by quantum circuits, so we will also sometimes speak of an ensemble of circuits. When we talk about the \emph{depth of an ensemble} (often just ``the depth") we mean the maximum over the circuit depths (conventionally defined) of all the circuits in an ensemble. 
\paragraph{Moment operators and group designs}
The \textit{k-wise twirl} (also called the $k$-th moment operator) $\Phi_{\mathcal{E}}^{(k)}$ of a ensemble of operators $\mathcal{E}$ is defined as
\begin{equation}
    \Phi_{\mathcal{E}}^{(k)}(A) \coloneqq \E_{U \sim \mathcal{E}} U\tn{k} A U\tnd{k},
\end{equation}
for a matrix $A \in \mathcal{L}(D^k)$. When it is clear from context what $k$ is, we will just write the twirl as~$\Phi_\mathcal{E}$, omitting the $k$ from the notation. We write $\Phi_{G}$ for the twirl over the Haar measure of group $G$.

\begin{definition}
    A unitary $k$-design over the group $G$ is an ensemble $\mathcal{E}$ of operators in $G$ such that
    \begin{equation}
        \Phi_{G}^{(k)}(A) = \Phi_{\mathcal{E}}^{(k)}(A),
    \end{equation}
    for every $A \in \mathcal{L}(\H\tn{k})$.
\end{definition}

The commutant of the \textit{$k$-th diagonal representation} $U \rightarrow U\tn{k}$ of $G$ (from now on we will just call this the commutant) is defined as $\Comm(G, k) \coloneqq \{A \in \mathcal{L}(\H\tn{k}) \mid A U\tn{k} = U\tn{k}A\}$. 

\begin{theorem}
    The Haar twirl is the projector on the commutant. More precisely, if $b_k$ is a basis of the commutant, then we can compute the Haar twirl as
    \begin{equation}
        \Phi_{G}^{(k)}(A) = \sum_{\sigma, \tau \in b_k} \Wg_G(\sigma,\tau) \tr(\sigma^\dagger A) \tau.
    \end{equation}
    Here the Weingarten function $\Wg_G(\sigma. \tau)$ is defined as the pseudo-inverse of the Gram matrix $\G_{\sigma, \tau} \coloneqq \tr(\sigma^\dagger \tau)$.
\end{theorem}
If we know the commutant of a group, we thus know how to compute the twirl. See~\cite{mele_introduction_2024} for more details on the Haar measure.

\paragraph{Approximate group designs}

As it is difficult to find unitary $k$-designs for higher values of $k$, we introduce different notions of approximate group designs. All the following definitions can be defined for any subgroup of the unitary group.

\begin{definition}[Relative-error approximate unitary $k$-designs]
    An ensemble $\mathcal{E}$ is an $\epsilon$-approximate relative-error unitary $k$-design over $G$ for $\epsilon > 0$ if
    \begin{equation}
        (1-\epsilon) \Phi_{G}^{(k)} \preceq \Phi_{\mathcal{E}}^{(k)} \preceq (1+\epsilon)\Phi_{G}^{(k)}.
    \end{equation}
    The notation $\Phi \preceq \Phi'$ means that $\Phi - \Phi'$ is a completely positive map.
\end{definition}

\begin{definition}[Additive-error approximate unitary $k$-designs]
    An ensemble $\mathcal{E}$ is an $\epsilon$-approximate additive-error unitary $k$-design over $G$ for $\epsilon > 0$ if
    \begin{equation}
        ||\Phi_{G}^{(k)} - \Phi_{\mathcal{E}}^{(k)}||_\lozenge \leq \epsilon. 
    \end{equation}
\end{definition}
Here, the diamond norm is defined as $||\Phi||_\lozenge = \max_{\{X; ||X||_1 \leq 1\}} ||(\Phi \otimes \one_n)(X)||_1$. A relative error $\epsilon$ always implies an additive error $2\epsilon$ \cite{brandao_local_2016}.
Even more, given a relative-error $\epsilon$-approximate unitary $k$-design $\mathcal{E}$, any quantum experiment that queries a unitary $U$ from $\mathcal{E}$ at most $k$ times outputs a state that is, on average, $2\epsilon$-close in trace distance to the output state obtained if $U$ were sampled from the Haar measure instead.
For an additive-error $\epsilon$-approximate unitary $k$-design, the same statement holds but then for non-adaptive quantum experiments~\cite{schuster_random_2025}. An additive error is thus a weaker requirement than a relative error.

\paragraph{Approximate state designs}
We will now introduce some notions of state designs.
\begin{definition}[Relative-error approximate state design]
    An ensemble $\mathcal{E}$ of pure states is an $\epsilon$-approximate relative-error state $k$-design over $G$ for $\epsilon > 0$ if
    \begin{equation}
        (1-\epsilon) \E_{\ket{\psi} \sim G} \ketbra{\psi}\tn{k} \preceq \E_{\ket{\psi} \sim \mathcal{E}} \ketbra{\psi}\tn{k} \preceq (1+\epsilon)\E_{\ket{\psi} \sim G} \ketbra{\psi}\tn{k}.
    \end{equation}
\end{definition}

\begin{definition}[Additive-error approximate state design]
    An ensemble $\mathcal{E}$ is an $\epsilon$-approximate additive-error state $k$-design for $\epsilon > 0$ if
    \begin{equation}
        \left|\left|\E_{\ket{\psi} \sim G} \ketbra{\psi}\tn{k} - \E_{\ket{\psi} \sim \mathcal{E}} \ketbra{\psi}\tn{k}\right|\right|_1 \leq \epsilon.
    \end{equation}
\end{definition}

Like with group designs, every quantum experiment that queries a state $\ket\psi$ at most $k$ times from a relative-error state $k$-design $\mathcal{E}$, outputs a state that is, on average, $2\epsilon$-close to the state we would have gotten if we would have sampled from the Haar random distribution. For additive-error a similar statement holds, but then for non-adaptive algorithms. 
If $\mathcal{E}$ is a relative-error unitary $k$-design, then $\{U\ket{0_n}\mid U \sim \mathcal{E}\}$ is a relative-error state $k$-design. A similar statement holds for additive error. State designs are thus a weaker notion than group designs.

\paragraph{Anti-concentration}
An even weaker notion of randomness is anti-concentration. 
As mentioned above, we here adopt a common misusage of this word in the random circuit community and mean that certain moments of the ensemble $\mathcal E$ approximate the Haar value over $G$.
For some groups $G$, not even the outcome distribution induced by Haar-random unitaries anti-concentrates in the conventional sense 
, and thus also not ensembles $\mathcal E$ obeying our definition. 
For the here studied groups, this is the case for the matchgate group.

Concretely, the definition we will use throughout this paper is the following:
\begin{definition}
    An ensemble $\mathcal{E}$ anti-concentrates up to order $k$ if
    \begin{equation}
        \E_{U\sim \mathcal{E}}\sum_{x \in \{0,1\}^n}|\bra{x} U \ket{0_n}|^{2k} \leq \alpha \cdot \E_{U\sim G}\sum_{x \in \{0,1\}^n}|\bra{x} U \ket{0_n}|^{2k}
    \end{equation}
    for some constant $\alpha \geq 1$.
\end{definition}
We note that if $\mathcal{E}$ forms a relative-error state $k$-design, then it anti-concentrates. 

\paragraph{Architectures}
Qubits are ordered in a \textit{architecture}. An architecture is described by a graph, where the vertices of the graph are the qubits, and an edge between two vertices indicates that a 2-qubit gate can be applied on those 2 qubits. Commonly studied architectures are $1\mathrm{D}$ and $2\mathrm{D}$ lattices (with open or closed boundary conditions), or all-to-all connectivity (the complete graph). 
We define the \textit{diameter} (often denoted as $\dia$) of an architecture as the largest possible distance between two vertices, e.g.\ the diameter of $1\mathrm{D}$ is $\dia = \Theta(n)$, for $2\mathrm{D}$ $\dia = \Theta(\sqrt{n})$, and for the all-to-all connectivity $\dia = \Theta(1)$.
Note, however, that many known constructions of designs in the all-to-all connectivity actually use architectures of higher diameter. For example, the construction of Clifford 3-designs in~\cite{schuster_random_2025} in $\Or(\log\log(n))$ depth implicitly uses an architecture of diameter $\Or(n/\log(n))$.

\paragraph{Circuit distribution classes}
We define a \textit{generalized random brickwork circuit} as a circuit ensemble that is constructed from independent $r$-local Haar-random gates, for some constant $r$.

\begin{definition}[Generalized random brickworks]\label{def:brickwork}
    Let $A_1,\ldots, A_L$ be partitions of $[n]$, where each $a_{i,j}\in A_{i}$ is a subset of $[n]$ of size at most some constant $r$, such that the qubits in $a_{i,j}$ are connected w.r.t to the underlying architecture. For each $a_{i,j}$ draw a unitary $U_{i,j}$ uniformly at random from the group $G$ on the qubits $a_{i,j}$ or apply the identity . Let $\mathcal{E}$ denote the ensemble of circuits generated by composing these unitaries in $L$ layers. We call an ensemble of this form a \textbf{fixed brickwork}. A \textbf{generalized random brickwork} is defined as a convex combination of (different) fixed brickworks.
\end{definition}

We emphasize that architecture-independent lower bounds derived for generalized random brickworks still hold under a significant relaxation of this definition, and we decided for Def.~\ref{def:brickwork} to simplify the presentation.
More precisely, we do not necessarily need that the local gates are uniformly distributed and their distributions do not have to be identical either (but they should be sampled independently).
Instead, our proofs require that the local gates preserve some single-qubit Pauli operator (up to a phase) with a finite probability or, more generally, that they can be replaced by such a distribution up to second moments (see also Sec.~\ref{sec:methods lower bounds} for a sketch of the argument).
This includes a vast majority of the random circuits considered in the literature.

See~Refs.~\cite{harrow_approximate_2023, dalzell_random_2022, haferkamp_random_2022, chen_incompressibility_2024}, and Corollary 1, first bullet point, in Ref.~\cite{schuster_random_2025} for examples of fixed brickworks, that are thus generalized random brickworks as well.
See Definition 1.4, first bullet point, in Ref.~\cite{chen_incompressibility_2024} for an example of a generalized random brickwork which is not a fixed brickwork.
See Ref.~\cite{suzuki_more_2024} (with CNOT as entangling gate) for an example that is a generalized random brickwork circuit under the relaxed definition .
See Refs.~\cite{metger_simple_2024, jiang_optimal_2020} and Corollary 1, second bullet point, in Ref.~\cite{schuster_random_2025} for examples of ensembles that are not generalized random brickworks.

When studying anti-concentration, we only consider generalized random brickworks. If not, then ensembles with constant depth can be found that anti-concentrate. For example, the circuit that applies a Hadamard gate $H$ on every qubit anti-concentrates for the unitary group.
We also obtain the strongest results for generalized random brickworks due to their clearly defined structure.

We define a \textit{locally independent ensemble} as an ensemble such that the local gates acting on a qubit are drawn independently from the local gates acting on qubits that are ``far away''.
\begin{definition}[Locally independent ensemble]\label{def:local}
    Let $\mathcal{E}$ be an ensemble of unitary circuits on a fixed architecture.
    Let $S_x(U) \coloneq \{U^x_1, \ldots , U^x_{m_{U}}\}_{U}$ be the set of local unitaries that act on qubit $x$ for a circuit $U$. If for every pair of qubits $x,y$ that have distance at least $\dia/2$ we have that $S_x(U)$ and $S_y(U)$ are independent variables when $U$ is sampled from $\mathcal{E}$, then $\mathcal{E}$ is \textbf{locally independent}.
\end{definition}
Note that the notion of a locally independent ensemble does not make any sense in the all-to-all architecture, as there we have that $\dia = 1$.
Almost all of the ensembles named in the previous paragraph are locally independent. Only~\cite{ghosh_random_2024} and \cite[Definition 1.4 first bullet]{chen_incompressibility_2024} are not, as these ensembles use the all-to-all architecture.

\section{Clifford group}
\label{section:clifford-group}

\subsection{Preliminaries}
The $n$-qudit Clifford group $\Cl_n$ is defined as the group of unitaries that normalize the Pauli group $\mathcal{P}_n$, i.e.~the group of all unitaries $U$ such that $U \mathcal{P}_n U^\dagger \subset \mathcal{P}_n$. 
Here, we will only consider qubits, in which case he Clifford group can equivalently be defined as the group generated by the CNOT-gate, the Hadamard gate $H$ and the phase gate $S$. 
The $n$-qubit Clifford group forms an exact unitary 3-design, but not a 4-design~\cite{zhu_clifford_2016}. We can thus directly extend the gluing lemma from~\cite{schuster_random_2025} to $k=3$. 

We note that the results from this section can be readily generalized to prime-dimensional qudits, in which case the Clifford group only forms a unitary 2-design, but not a 3-design \cite{zhu_multiqubit_2017}.
The analogous statements of our theorems would then concern bounds on the circuit depth of Clifford 3-designs.

We are interested in the diagonal representation of the Clifford group, $U \rightarrow U^{\otimes k}$. The commutant of this representation can be described using \textit{Stochastic Lagrangian subspaces}~\cite{gross_schurweyl_2021}.

\begin{definition}
    The subspace $T \subset \Z_2^{2k}$ is a Stochastic Lagrangian subspace if the following three properties hold.
    \begin{enumerate}
        \item For every $(x,y) \in T$, we have that $x \cdot x = y \cdot y \mod 4$
        \item $T$ has dimension $k$
        \item $T$ is stochastic, which means that $(1,\ldots,1) \in T$.
    \end{enumerate}
    We define $\Sigma_{k,k}$ as the set of all Stochastic Lagrangian subspaces.
\end{definition}

We note that condition 1 implies that $T$ has dimension at most $k$. Condition 2 could thus equivalently be stated as $T$ having the largest possible dimension. The commutant of the diagonal representation can now be described as follows.
\begin{theorem}[\cite{gross_schurweyl_2021}]
    If $n \geq k-1$, then the commutant of the diagonal representation is spanned by the operators $R(T)=r(T)\tn{n}$ for $T \in \Sigma_{k,k}$, where
    \begin{equation}
        r(T) \coloneqq \sum_{(x,y) \in T} \ketbra{x}{y}.
    \end{equation}
\end{theorem}
The Clifford group is a subgroup of the unitary group. This means that the commutant of the Clifford group should at least include the commutant of the unitary group, which we know to be the permutations. Indeed, for $\sigma \in S_k$ a permutation, we see that $x \cdot x = (\sigma x) \cdot (\sigma x) \mod 4$. Hence, the space $T_\sigma=\{(\sigma x, x) \mid x \in \Z_2^k\}\in \Sigma_{k,k}$. In the following, we will hence view $S_k$ as a subset of $\Sigma_{k,k}$ via this identification, i.e., $S_k \subseteq \Sigma_{k,k}$. Indeed, for $k\leq 3$, $S_k = \Sigma_{k,k} $. For $k=4$ however, there exist other subspaces and so $S_k \subset \Sigma_{k,k}$. An example of such a space is $T_4$ with representation
\begin{equation}\label{eq:T4}
    R(T_4) = 2^{-n}\sum_{P \in \Pa_{n}} P^{\otimes 4},
\end{equation}
where the sum runs over all $n$-qubit Pauli matrices.

\subsection{Clifford group designs}
In this section we prove a lower bound for the circuit depth of additive-error Clifford (unitary) designs based on lightcone arguments. This bound is remarkably general, but depends on the underlying architecture. 

For generalized random brickwork circuits (see Def.~\ref{def:brickwork}), we obtain a much stronger bound, namely $\Omega(n)$, independent of the architecture.
This intriguing result follows from a different argument exploiting the structure of such circuits and the discreteness of the Clifford group, and is derived for relative-error state designs in the next section.

\begin{theorem}\label{theorem: clifford rel unitary}
    Let $\mathcal{E}$ be an approximate Clifford unitary $k$-design (with $k\geq 4$) with additive error $\varepsilon < 1/4$. Then the lightcone of at least one Clifford in the ensemble must cover a constant fraction of the qubits. 
    This implies $\mathcal{E}$ must have circuit depth $d = \Omega(n)$ in $1\mathrm{D}$ circuits and $d = \Omega(\log n)$ in all-to-all connected circuits.
\end{theorem}
    
\begin{proof}
    To prove the theorem, we will show that any Clifford unitary with light-cone smaller than the system size can be distinguished from a Haar-random Clifford unitary using four parallel queries and constant-depth quantum and classical operations.
    
    Let $C$ be a short-depth Clifford unitary and $\chi_n = \frac{1}{8^n} R(T_4) = \frac{1}{16^n} \sum_{P \in P_n} P^{\otimes 4}$ be the mixed state with stabilizers $P^{\otimes 4}$ for all Pauli operators $P$.
    The state $\chi_n$ can be prepared in constant-depth since it is a tensor product of $n$ 4-qubit stabilizer states, $\chi_n = \chi_1^{\otimes n}$.
    Our distinguishing protocol proceeds as follows.
    We first perturb $\chi_n$ by a single-qubit Pauli operator $Z_1$ on the first qubit (taking some arbitrary but fixed order of the qubits) of the first copy.
    We then apply $C$ to all four copies in parallel.
    Finally we choose $i\in[n]$ uniformly at random and we measure the expectation value of the local tensor power Pauli operator $Z_i^{\otimes 4}$.

    This procedure yields an expectation value,
    \begin{equation*}
        \mathbb{E}_{i\sim [n]}\tr\left[ Z_i^{\otimes 4} C^{\otimes 4} (Z_1 \otimes \mathbbm{1}^{\otimes 3}) \chi_n (Z_1 \otimes \mathbbm{1}^{\otimes 3}) C^{\dagger,\otimes 4} \right]
        =
        \mathbb{E}_{i\sim [n]}\tr\left[ Z_i^{\otimes 4} (C Z_1 C^\dagger \otimes \mathbbm{1}^{\otimes 3}) \chi_n (C Z_1 C^\dagger \otimes \mathbbm{1}^{\otimes 3})  \right].
    \end{equation*}
    In the second expression, we use that $\chi_n = C^{\dagger, \otimes 4} \chi_n C^{\otimes 4}$ since $\chi_n$ commutes with every Clifford unitary $C$.
    We then cancel the Clifford unitaries on the final three copies.
  Now imagine all Cliffords  $C$ in $\mathcal{E}$ have a lightcone of size less than $n/4$. Then the resulting expectation value is larger than $1/2$ because $Z_i^{\otimes 4}$ and $C Z_1 C^\dagger$ commute whenever $i$ is outside of the lightcone and $Z_i^{\otimes 4}$ is a stabilizer of $\chi_n$ (which is true for a $3/4$'th fraction of the qubits). For the qubits inside the lightcone (a less than $1/4$'th fraction of the qubits) the expectation value cannot be more negative that $-1$, making the overall expectation value larger than $1/2$.
    Meanwhile, in the Haar-random Clifford unitary ensemble, the expectation value of the expectation value is exponentially small in $n$.
    This implies that $\lVert \Phi^{(4)}_\mathcal{E} - \Phi^{(4)}_{\text{Cl}_H} \rVert_\lozenge \geq 1/4$ for any Clifford ensemble $\mathcal{E}$ of Cliffords whose light-cone is of size less than $n/4$.
    This completes the proof.
\end{proof}

\subsection{Clifford state designs: relative error}\label{section: cliffords rel state}
\begin{theorem}\label{theorem:clifford_rel_state}
    Let $\mathcal{E}$ be an ensemble such that $\Phi_\mathcal{E} \left(\ketbra{0^n}\tn{4}\right)$ is a relative-error $\epsilon$-approximate Clifford state $4$-design of depth $d$ for $\epsilon < 1$. The following 2 statements hold:
    \begin{enumerate}
        \item If the architecture has diameter $\dia$ and $\mathcal{E}$ is locally independent (Def.~\ref{def:local}), then $d = \Omega(\dia)$.
        \item If $\mathcal{E}$ is a generalized random brickwork (Def.~\ref{def:brickwork}), then $d = \Omega(n)$.
    \end{enumerate}
\end{theorem}
\begin{proof}
    Let $U$ be a Clifford unitary and consider $R(T_4)$ as in Eq.~\eqref{eq:T4}.
    We find:
    \begin{align}
        \tr&\left[(Z_i \otimes \one\tn{3}) R(T_4) (Z_i \otimes \one\tn{3}) U\tn{4}(\ketbra{0000})U\tnd{4}\right] \\
        &= \sandwich{0000}{(U Z_i U^\dagger \otimes \one^{\otimes 3}) R(T_4) (U Z_i U^\dagger \otimes \one^{\otimes 3}) }{0000} \\
        &= \frac{1}{2^n} \sandwich{0}{(U Z_i U^\dagger) (\one+Z)^{\otimes n} (U Z_i U^\dagger) }{0} \\
        &= \bra{00} U\tn{2} Z_i\tn{2} U\tnd{2} \ket{00}. \label{eq:Clifford rel state Pauli EV}
    \end{align}
    In step two, we performed the partial contraction on the last three systems and used that $\sandwich{0}{P}{0} = 1$ if $P$ is a diagonal Pauli and zero else.
    In the last step, we resumed the stabilizers of the all-zero state.
    If $U$ is sampled from the Haar measure on the Clifford group, we get
    \begin{equation}\label{eqn: clifford random small}
        \mathbb{E}_{U\sim\text{Cl}_H} \bra{00} U\tn{2} Z_i\tn{2} U\tnd{2} \ket{00} 
        = \frac{1}{2^{n}(2^n + 1)}\tr\left[Z_i\tn{2} (\one + \F) \right]
        = \frac{1}{2^n + 1}.
    \end{equation}
    For sampling from an ensemble $\mathcal{E}$ of Clifford circuits of depth $d$, we find:
    \begin{equation}\label{eqn: Clifford lightcone state}
    \begin{split}
        \E_{U \sim \mathcal{E}}\bra{00} U\tn{2} Z_i\tn{2} U\tnd{2} \ket{00} &= \E_{U \sim \mathcal{E}}\bra{00}_{l_i} U_{l_i}\tn{2} Z_i\tn{2} U_{l_i}\tnd{2} \ket{00}_{l_i},
    \end{split}
    \end{equation}
    where the subscript $l_i$ indicates that the unitary only acts on the qubits inside the lightcone starting from qubit $i$. As $\mathcal{E}$ is a 4-design, Eq.~\eqref{eqn: Clifford lightcone state} is within relative error $\epsilon$ to the Haar random value $\frac{1}{2^n + 1}$.
    
    1.~Choose $i$ and $j$ as two qubits that have distance $\dia$ from each other. If $\mathcal{E}$ has depth $d < \dia/4$, the following equations need to hold:
    \begin{equation}\label{eqn: Clifford rel state contradiction}
    \begin{split}
        \E_{U \sim \mathcal{E}}\bra{00}_{l_i} U_{l_i}\tn{2} Z_i\tn{2} U_{l_i}\tnd{2} \ket{00}_{l_i}
        &\leq (1+\epsilon) \frac{1}{2^n + 1}\\
        \E_{U \sim \mathcal{E}}\bra{00}_{l_j} U_{l_j}\tn{2} Z_j\tn{2} U_{l_j}\tnd{2} \ket{00}_{l_j}
        &\leq (1+\epsilon) \frac{1}{2^n + 1}\\
        \E_{U \sim \mathcal{E}}\bra{00} U\tn{2} (Z_i Z_j)\tn{2} U\tnd{2} \ket{00}
        &= \E_{U \sim \mathcal{E}} \bra{00}_{l_i} U_{l_i}\tn{2} Z_i\tn{2} U_{l_i}\tnd{2} \ket{00}_{l_i} \\ 
        & \qquad \times \bra{00}_{l_j} U_{l_j}\tn{2} Z_j\tn{2} U_{l_j}\tnd{2} \ket{00}_{l_j} \geq (1-\epsilon) \frac{1}{2^n + 1},
    \end{split}
    \end{equation}
    where we used that $d < \dia/4$, so that the lightcones starting from $i$ and $j$ do not intersect. We directly see that these equations contradict the locally independent assumption. This proves point~1.
    
    2.~Let $\mathcal{E}$ be a generalized random brickwork (see Def.~\ref{def:brickwork}) of depth $d$.
    By construction Eq.~\eqref{eq:Clifford rel state Pauli EV} is non-negative, hence we can lower bound its average by the probability that $\bra{00} U\tn{2} Z_1\tn{2} U\tnd{2} \ket{00} = 1$ over all circuit instances $U$.
    This probability can be equivalently phrased as the probability that $U Z_1 U^\dagger$ is a diagonal Pauli operator.
    Here, we lower bound this probability exploiting that the local Clifford gates in every layer are uniformly distributed -- we however note that such an argument can likely be made for much more general distributions (cf.~the comment below Def.~\ref{def:brickwork}).
    Because each brick has constant-size support, there is a constant probability $p$ that the first brick maps $Z_1$ to a weight-one Pauli operator (not necessarily supported on the first qubit anymore). 
    More precisely, if the support is of size $r$, then $p=3r/(4^r-1)$.
    The probability that the bricks in the remaining layers preserve the weight is also precisely $p$.
    Since the ensemble is invariant under local Cliffords, the probability that the final Pauli operator after $d$ layers is diagonal is $p^d/3$.
    Hence, we have $\bra{00} U\tn{2} Z_1\tn{2} U\tnd{2} \ket{00} = 1$ with probability is at least $p^d/3$, resulting in the lower bound
    \begin{equation}
    \begin{split}
        \EE_{U\sim\mathcal E}\bra{00} U\tn{2} Z_i\tn{2} U\tnd{2} \ket{00}
        \geq \frac{1}{2^{\Or(d)}}.
    \end{split}
    \end{equation}
    This, together with Eq.~\eqref{eqn: clifford random small}, implies a lower bound of $d = \Omega(n)$ for $\mathcal{E}$ to be a relative-error Clifford state 4-design.
\end{proof}
We believe the assumption of local independence in the first point of the theorem can be dropped, but we do not know how to prove this. In particular, our proof cannot easily be generalized. Define for instance $\mathcal{E}$ as the ensemble that draws~$\one$ with probability $\frac{1}{2^n+1}$, and draws $H\tn{n}$ with probability $1 - \frac{1}{2^n+1}$. Then all three values in Eq.~\eqref{eqn: Clifford rel state contradiction} are exactly $\frac{1}{2^n(2^n+1)}$, the Haar random value. The ensemble $\mathcal{E}$ is however not locally independent, and it does not define a Clifford state design.

\subsection{Clifford state designs: additive error}
Above we have shown that the Clifford group does not allow for short-depth state $k$-designs (for $k>3$) in the relative-error model. However, somewhat surprisingly, it turns out to be possible to construct short-depth state $k$-designs for $k=4$ and $k=5$ with small (i.e. $\Or(2^{-n})$) additive error. To prove this we need the following lemma that shows that for $k\leq 5$ Clifford unitaries applied to a PPT input state (for all possible partial transposes of $n$ qubits) generate $U(2^n)$ state $k$-designs with small additive error:
\begin{lemma}\label{lem:clifford_state_ppt}
    Consider a state $\rho$ on $\mathcal{H}^{\otimes k}$. Suppose that $\rho$ obeys the PPT condition on each copy of $\mathcal{H}$, i.e.~$\rho^{\Gamma_i} \geq 0$ for all $i = 1,\ldots,k$, where $\Gamma_i$ denotes the partial transpose on the $i$-th copy.
    Then, if $k\leq 5$, the Clifford twirl of $\rho$ is equal to the Haar unitary twirl of $\rho$ up to additive error:
\begin{equation}
\norm{\Phi_{\mathrm{Cl}_{H}}^{(k)}\left(\rho\right)-\Phi_{U_{H}}^{(k)}\left(\rho\right)}_1
\leq \Or(2^{-n}).
\end{equation}
\end{lemma}
\begin{proof}
The lemma is trivial for $k\leq 3$. For $k=4,5$, the set of stochastic Lagrangian subspaces is given by $\Sigma_{k,k}=S_k \cup \overline{S_k}$ where $\overline{S_k}=\Sigma_{k,k}\setminus S_k$ denotes the complement of $S_k$. Hence, the commutant of the Clifford group is spanned by the operators 
\begin{equation*}
   \{R(T)\}_{T \in \Sigma_{k,k}}= \{R(\pi)\: |\: \pi \in S_k\}\cup \{R(T) \; |\; T\in \overline{S_k} \}.
\end{equation*}
 Specifically for $k=4,5$, all $T \in \overline{S_k}$, $R(T)$ can be written as $R(T)=R(T_k) R(\pi)$ for a suitable permutation $\pi \in S_k$ with $T_4$ defined in Eq. \ref{eq:T4} and $T_5 = T_4\otimes I$. The $R(T)$ for $T\in \overline{S_k}$ are not unitary, and we have $\norm{R(T)}_1 = 2^{n(k-1)}$. We will now show that for each of the operators $R(T_k)R(\pi)$ there is a partial transpose such that $\norm{(R(T_k)R(\pi))^{\Gamma_{i}}}=1$. Note that this is certainly true for $R(T_k)$. It can be checked by direct computation that $R(T_k)^{\Gamma_i}$ for $i\in\{1,2,3,4\}$ is unitary. Now choose $i\in \{1,2,3,4\}$ such that $\pi(i)\neq 5$ (which can always be done) and apply $\Gamma_i$. Up to a permutation $\sigma$ (which does not increase the operator norm of $R(T_k)R(\pi))^{\Gamma_{i}}$ we have
\begin{equation}
    (R(T_k)R(\pi))^{T_{i}}R(\sigma) = 2^{n}\sum_{P \in \Pa_{n}} (I\otimes)^{\delta_{k=5}} P\tn{2}\otimes \big(I\otimes P \dens{\Phi}I\otimes P\ct\big)
\end{equation}
where $\Phi$ is the maximally entangled state. This operator factorizes over $n$ and can be checked numerically to have operator norm $1$. As a consequence, we can show that for $k\leq 5$,
\begin{equation}
    \left|\tr\left(R(T)\rho\right)\right| \leq 1, \quad \forall\, T\in \Sigma_{k,k}. \label{eq:r-t-ppt-bound}
\end{equation}
To see this, note that using Hölder's inequality, we have
\begin{equation}
    \left|\tr\left(R(T)\rho\right)\right|=\left|\tr\left(R(T)^{\Gamma_{i}}\rho^{\Gamma_{i}}\right)\right|\leq\left\Vert R(T)^{\Gamma_{i}}\right\Vert _{\infty}\left\Vert \rho^{\Gamma_{i}}\right\Vert _{1}=1.
\end{equation}

\noindent Now, we can bound the distance of the twirls directly, as follows

\begin{align}
&\left\Vert \Phi_{\mathrm{Cl}_{H}}^{(k)}\left(\rho\right)-\Phi_{U_{H}}^{(k)}\left(\rho\right)\right\Vert _{1} \\
=& \norm{\sum_{T,T'\in\Sigma_{k,k}}\Wg_{\mathrm{Cl}}^{\left(k\right)}\left(T,T'\right)\tr\left(R(T)\rho\right)R\left(T'\right)-\sum_{T,T'\in S_{k}}\Wg_{\mathrm{U}}^{\left(k\right)}\left(T,T'\right)\tr\left(R(T)\rho\right)R\left(T'\right) }_1 \\
  \leq&\sum_{T,T'\in\Sigma_{k,k}}\left|\Wg_{\mathrm{Cl}}^{\left(k\right)}\left(T,T'\right)-\Wg_{\mathrm{U}}^{\left(k\right)}\left(T,T'\right)\right|\left|\tr\left(R(T)\rho\right)\right|\left\Vert R\left(T'\right)\right\Vert _{1}\\
  \leq&\sum_{T,T'\in\Sigma_{k,k}}\left|\Wg_{\mathrm{Cl}}^{\left(k\right)}\left(T,T'\right)-\Wg_{\mathrm{U}}^{\left(k\right)}\left(T,T'\right)\right|\left\Vert R\left(T'\right)\right\Vert _{1} \label{eq:clifford_state_ppt_weingarten_bound}
\end{align}
where we slightly abuse notation by setting $\Wg_{\mathrm{U}}^{\left(k\right)}\left(T,T'\right)=0$
if $T\not\in S_{k}$ or $T'\not\in S_{k}$. In the last step, we used Eq. \ref{eq:r-t-ppt-bound}. Next, we split the sum over $\Sigma_{k,k}$.

In particular, we can upper bound this sum by four terms by choosing $T, T'$ in $S_k$ or $\overline{S_k}$. For visual help, we will write $\pi\in S_k$ and $T \in \overline{S_k}$. In order we have
\begin{equation}
(1) \leq \sum_{\pi,\pi' \in S_k} 2^{nk} |\Wg_{\Cl}^{(k)} (\pi, \pi') - \Wg_{\U}^{(k)}(\pi, \pi') | \leq \Or (2^{-n})
\end{equation}
where we used that $\norm{R(\pi)}_1 = 2^{nk}$ for all $\pi\in S_k$ and $|\Wg_{\Cl}^{(k)}(\pi,\pi') - \Wg_{\U}^{(k)}(\pi,\pi')|\leq \Or(2^{-n(k+1)})$ (via a straightforward Weingarten calculation, see also~Ref.~\cite{helsen_thrifty_2023}).
Similarly we can see that 
\begin{equation}
(2) \leq \sum_{\pi \in S_k, T'\in \overline{S_k}} |\Wg_{\Cl}^{(k)}(\pi,T')|  \norm{R(T')}_1 = \Or(2^{-n}),
\end{equation}
since $|\Wg_{\Cl}^{(k)}(\pi,T')| \leq 2^{-(k+1)n}$ and $\norm{R(T')}_1\leq 2^{n(k-1)}$. Moving on we have
\begin{equation}
(3)\leq  \sum_{T \in \overline{S_k},\, \pi'\in {S}_k} |\Wg_{\Cl}^{(k)}(T,\pi')|  \norm{R(\pi')}_1 = \Or(2^{-n}),
\end{equation}
since again $|\Wg_{\Cl}^{(k)}(T,\pi')| \leq 2^{-(k+1)n}$. Finally we have 
\begin{equation}
(4)\leq  \sum_{T \in \overline{S_k}, \, T'\in \overline{S_k}} |\Wg_{\Cl}^{(k)}(T,T')| \norm{R(T')}_1 \leq \Or(2^{-n})
\end{equation}
which follows from $\norm{R(T')}_1 =2^{n(k-1)}$ and $|\Wg_{\Cl}^{(k)}(T,T')| \leq O(2^{-nk})$. This completes the proof. 
\end{proof}

We note that Eq.~\eqref{eq:clifford_state_ppt_weingarten_bound} can be computed exactly using Weingarten calculus.
This allows to bound the constant in Lemma \ref{lem:clifford_state_ppt} by 1/4 (for $n > 2$) \cite{turkeshi_PC_2025}.

From this construction and the gluing lemma due to \cite{schuster_random_2025} we get in a straightforward way that the superblock construction applied to the all-zero initial state is an additive error state design.
\begin{corollary}\label{cor:clifford add state}
Consider the two-layer superblock brickwork circuit composed of uniformly random Clifford circuits  acting on $2\xi$ qubits, applied to the all-zero initial state.
There is a constant $C\geq1$ such that for any $\epsilon = \Omega(2^{-n})$, $k=4,5$, and $\xi \geq \log_2(C n/\epsilon)$, this ensemble forms an $\epsilon$-approximate additive-error Clifford state $k$-design with depth $\Or(\xi)=\Or(\log(n/\epsilon))$.
\end{corollary}
\begin{proof}
Let $C'$ the constant in Lemma~\ref{lem:clifford_state_ppt}.
As the input state $\ket{0}\tn{k}$ is separable, we can replace the uniformly random Cliffords in the first layer with Haar random unitaries (using \cref{lem:clifford_state_ppt}), incurring a total additive error $\leq C' n 2^{-2\xi} \leq C' C^{-2} \epsilon^2  k^{-2} n^{-1} \leq C'C^{-1}\epsilon$ (as there are $n/2\xi\leq n$ many gates). 
Similarly, since the state after this first layer of unitaries is still separable, we can replace the uniformly random Cliffords in the second layer by Haar-random unitaries, incurring the same amount of error again. 
We now have a superblock circuit with Haar-random unitaries for which we can invoke the gluing lemma of Ref.~\cite{schuster_random_2025} to bound the relative error of the whole circuit (to a global Haar-random unitary).
More precisely, this error is computed in Ref.~\cite[App.~B.3.4]{schuster_random_2025}.
Inserting $k\geq 4$ in the derivation of \cite[Eq.~(B.31)]{schuster_random_2025}, we find that $\xi\geq\log_2(C nk^2/\epsilon)$ implies a relative error $\leq \epsilon/(2C)$.
This converts into an additive error $\epsilon/C$ \cite[Lem.~3]{brandao_local_2016}.
Finally, exchanging the global Haar unitary back to a random Clifford yields another additive error $\leq C' 2^{-n}$.
The total error is then $\frac{C' + 1}{C}\epsilon + C' 2^{-n}$.
If the target error is larger than $C'2^{-n}$, then choosing $C=C'+1$ and redefining $\epsilon$ to be the difference between the target error and $C'2^{-n}$ shows that the ensemble is the required $\epsilon$-design. Since one can implement arbitrary Clifford circuits in $1$D in depth $O(\xi)$ the overall depth of the circuit is $O(\log(nk/\epsilon)$, which is what we set out to prove. 
\end{proof}

\section{Orthogonal group}
\subsection{Preliminaries}

The orthogonal group $\O(D)$ is defined as the set of all unitary matrices that have real entries, or, equivalently, the set of real matrices $O$ such that $O O^T = O^T O = \one$. It is known that orthogonal unitary $k$-designs can be constructed in depth linear in $n$ and polynomial in $k$~\cite{odonnell_explicit_2023}.

We quickly review orthogonal Weingarten calculus here. The commutant of the orthogonal group is a representation of the \textit{Brauer algebra}.

Let $M_k$ denote the set of all possible pairings of $2k$ vertices. If we number the vertices $1,2,3,4$, the elements of $M_2$ are thus $\{\{1,2\}, \{3,4\}\}, \{\{1,3\}, \{2,4\}\}$ and $\{\{1,4\}, \{2,3\}\}$. We define a multiplication $\sigma \cdot \tau$ for $\sigma, \tau \in M_k$ as concatenating the last $k$ vertices of $\sigma$ with the first $k$ vertices of $\tau$. For $k = 2$ this means for example that $\sigma \cdot \tau \coloneqq \{\{1,2\}, \{3,4\}\} \cdot \{\{1,3\}, \{2,4\}\} = \{\{1,2\}, \{3,4\}\}$, or, in graphical notation

\begin{equation}
\left(
\begin{tikzpicture}[baseline={(current bounding box.center)}, scale=0.5]
\node[draw, circle, scale = 0.8] (v1) at (0, 0) {1};
\node[draw, circle, scale = 0.8] (v2) at (0, -2) {2};
\node[draw, circle, scale = 0.8] (v3) at (2, 0) {3};
\node[draw, circle, scale = 0.8] (v4) at (2, -2) {4};
\draw (v1) -- (v2);
\draw (v3) -- (v4);
\end{tikzpicture}
\right)
\cdot
\left(
\begin{tikzpicture}[baseline={(current bounding box.center)}, scale=0.5]
\node[draw, circle, scale = 0.8] (v1) at (0, 0) {1};
\node[draw, circle, scale = 0.8] (v2) at (0, -2) {2};
\node[draw, circle, scale = 0.8] (v3) at (2, 0) {3};
\node[draw, circle, scale = 0.8] (v4) at (2, -2) {4};
\draw (v1) -- (v3);
\draw (v2) -- (v4);
\end{tikzpicture}
\right)
=
\left(
\begin{tikzpicture}[baseline={(current bounding box.center)}, scale=0.5]
\node[draw, circle, scale = 0.8] (v1) at (0, 0) {1};
\node[draw, circle, scale = 0.8] (v2) at (0, -2) {2};
\node[draw, circle, scale = 0.8] (v3) at (2, 0) {3};
\node[draw, circle, scale = 0.8] (v4) at (2, -2) {4};
\draw (v1) -- (v2);
\draw (v3) -- (v4);
\end{tikzpicture}
\right).
\end{equation}

However, sometimes it happens that by concatenating elements from $M_k$, we get one or more closed loop. For every closed loop, we multiply by some scalar $D$. An example of this for $k=2$ is

\begin{equation}
\left(
\begin{tikzpicture}[baseline={(current bounding box.center)}, scale=0.5]
\node[draw, circle, scale = 0.8] (v1) at (0, 0) {1};
\node[draw, circle, scale = 0.8] (v2) at (0, -2) {2};
\node[draw, circle, scale = 0.8] (v3) at (2, 0) {3};
\node[draw, circle, scale = 0.8] (v4) at (2, -2) {4};
\draw (v1) -- (v2);
\draw (v3) -- (v4);
\end{tikzpicture}
\right)
\cdot
\left(
\begin{tikzpicture}[baseline={(current bounding box.center)}, scale=0.5]
\node[draw, circle, scale = 0.8] (v1) at (0, 0) {1};
\node[draw, circle, scale = 0.8] (v2) at (0, -2) {2};
\node[draw, circle, scale = 0.8] (v3) at (2, 0) {3};
\node[draw, circle, scale = 0.8] (v4) at (2, -2) {4};
\draw (v1) -- (v2);
\draw (v3) -- (v4);
\end{tikzpicture}
\right)
= D 
\left(
\begin{tikzpicture}[baseline={(current bounding box.center)}, scale=0.5]
\node[draw, circle, scale = 0.8] (v1) at (0, 0) {1};
\node[draw, circle, scale = 0.8] (v2) at (0, -2) {2};
\node[draw, circle, scale = 0.8] (v3) at (2, 0) {3};
\node[draw, circle, scale = 0.8] (v4) at (2, -2) {4};
\draw (v1) -- (v2);
\draw (v3) -- (v4);
\end{tikzpicture}
\right).
\end{equation}
The Brauer algebra is now defined as the set of formal linear combinations of pairings:
\begin{equation}
\mathcal{B}_k(D) = \left\{\sum_{\sigma \in M_k} c_\sigma \sigma \mid c_\sigma \in \Z(D) \right\}.
\end{equation}

For an element $\sigma \in M_k$, let us define a map $f_{\sigma}:\ [2k] \rightarrow [2k]$ which sends a vertex $x$ to the vertex $y$ that it is paired with under $\sigma$. We can now define a representation $R_{\O}:\ \mathcal{B}_k(D) \rightarrow \mathcal{U}(D^k)$ of the Brauer algebra. For $\sigma \in M_k$ we define
\begin{equation}
R_{\O}(\sigma) \coloneqq \sum_{i_1,\ldots i_{2k}} \left(\prod_{j=1}^{2k} \delta_{i_j, i_{f_\sigma(j)}}\right) \ketbra{i_1\ldots i_k}{i_{k+1}\ldots i_{2k}}.
\end{equation}

We chose $D$ to be the dimension of one copy of our Hilbert space. If we define $\sigma = \{\{1,2\}, \{3,4\}\}$, we indeed get that
\begin{equation}
R_{\O}(\sigma) \cdot R_{\O}(\sigma) = \left(\sum_{i_1, i_{2}} \ketbra{i_1 i_1}{i_{2} i_{2}}\right) \left( \sum_{i_1, i_{2}} \ketbra{i_1 i_1}{i_{2} i_{2}}\right) = D \sum_{i_1, i_{2}} \ketbra{i_1 i_1}{i_{2} i_{2}} = D \cdot R_{\O}(\sigma).
\end{equation}

We note that all permutations are in the representation of the Brauer algebra.

The elements of $M_k$ under representation $R_O$ form a basis of the commutant of the orthogonal group~\cite{aharonov2021quantum}. We can thus use orthogonal Weingarten calculus to compute the moments of the orthogonal group as
\begin{equation}
\Phi_{\O}^{(k)}(A) = \sum_{\sigma, \tau \in M_k} \Wg_{\O}(\sigma, \tau) \tr(A \sigma^T) \tau,
\end{equation}
where $\Wg_O$ is the inverse of the orthogonal Gram matrix $(G_O)_{\sigma,\tau} = \tr(\sigma\tau^T)$. We note that, as $M_k$ does not form a group and in particular does not have inverses, the orthogonal Gram and Weingarten matrix do not have a constant row sum.

It will be useful to have a closer look at the case of $k=2$. We define $\ket{\Omega_n} \coloneqq \frac{1}{2^{n/2}}\sum_{x\in \{0,1\}^n} \ket{xx}$ as the maximally entangled state on two copies of $(\C^2)^{\otimes n}$. The basis for the commutant will now consist of the identity operator $\one$, the Flip operator $\mathbb{F}$ and $\Pi_o \coloneqq 2^n \ketbra{\Omega_n}$. Computing the Gram matrix and its inverse we get

\begin{align}
    G_{\O} = D \begin{pmatrix} D & 1 & 1 \\ 1 & D & 1 \\ 1 & 1 & D \end{pmatrix},\ 
    \Wg_{\O} = \frac{1}{D(D-1)(D+2)} \begin{pmatrix} D+1 & -1 & -1 \\ -1 & D+1 & -1 \\ -1 & -1 & D+1 \end{pmatrix} \,.
\end{align}
Hence,
\begin{equation}\label{eqn: orthogonal 2 twirl}
\begin{split}
    \EE_{U\sim \O} U\tn{2} A U\tnd{2} 
    &=
    \frac{1}{D(D-1)(D+2)}
    \Big(
    \left[ (D+1)\tr(A) - \tr(A\F) - \tr(A\Pi_o) \right] \one \\
    &+
    \left[ -\tr(A) + (D+1) \tr(A\F) - \tr(A\Pi_o) \right] \F \\
    &+
    \left[ -\tr(A) - \tr(A\F) + (D+1) \tr(A\Pi_o) \right] \Pi_o
   \Big).
\end{split}
\end{equation}

\subsection{Orthogonal group designs}\label{section:matrix_orthogonal}
The following argument was already given by~\cite{schuster_random_2025}, we repeat the main idea of the proof here for completeness.
For every orthogonal matrix $O\in\O(2^n)$, we have $O\otimes O \ket{\Omega_n} = OO^T\otimes\one\ket{\Omega_n}=\ket{\Omega_n}$, i.e.~$\ket{\Omega_n}$ is an invariant state under the tensor square representation.
We then study the action on the initial state $Z_1\otimes\one \ket{\Omega_n}$ where $Z_1$ denotes the Pauli $Z$ operator on the first qudit.
If we act with a Haar-random $O\in\O(2^n)$, we obtain the state $(O Z_1 O^T\otimes\one) \ket{\Omega_n}$ which will be highly entangled between the first $n$ qudits with high probability.
In contrast, if we act with a local random $\O$-circuit, then all gates outside the lightcone of the first qudit will cancel.
In particular, if the lightcone does not encompass all qudits, we are left with local maximally entangled states $\ket{\Omega_{\overline{l}}}$ in the complement of the lightcone $l$.
These are unentangled with the remaining qubits and can be detected, for instance by measuring their fidelity.
Hence, short local random $\O$-circuits cannot reproduce this second (psd) moment correctly.
In a $1\mathrm{D}$ nearest-neighbor architecture, they need at least $\Omega(n)$ depth, and in an all-to-all architecture $\Omega(\log(n))$ depth is required. We provide a formal proof here:
\begin{theorem}\label{theorem: orthogonal rel group design}
    Let $\mathcal{E}$ be an approximate orthogonal unitary $k$-design with additive error $\varepsilon < 1/4$. Then at least one orthogonal unitary must have a lightcone covering a $1/4$'th fraction of the qubits. 
    This implies $\mathcal{E}$ must have circuit depth $d = \Omega(n)$ in $1\mathrm{D}$ circuits and $d = \Omega(\log n)$ in all-to-all connected circuits for any $k \geq 2$.
\end{theorem}

\begin{proof}
    To prove the theorem, we will show that any orthogonal unitary with light-cone size smaller than $n/4$ can be distinguished from a Haar-random orthogonal unitary using two parallel queries and constant-depth quantum and classical operations.
    
    Let $O$ be a short-depth orthogonal unitary and $\ket{\Omega_n}$ the EPR state on two copies.
    Our distinguishing protocol proceeds as follows.
    We first perturb the EPR state by a single-qubit Pauli operator $Z_1$ on the first qubit (defined in a fixed but arbitrary order) of the first copy.
    We then apply $O$ to both copies in parallel.
    Finally, we choose $i\in [n]$ uniformly at random and measure the expectation value of the local EPR projector, $P^{(i)} = \mathbbm{1}^{\otimes 2(n-1)} \otimes \dyad*{\Omega^{(i)}_1}$, where $\ket*{\Omega^{(i)}_1}$ is the EPR state on the $i$'th qubit of the two copies.

    This procedure yields an expectation value,
    \begin{equation*}
        \mathbb{E}_{i\sim [n]}\bra{\Omega_n}(Z_1\otimes \one) (O^T \otimes O^T) P^{(i)}(O\otimes O)(Z_1\otimes \one) \ket{\Omega_n}
        =
        \mathbb{E}_{i\sim [n]} \bra{\Omega_n}(O Z_1 O^T \otimes \one)  P^{(i)}(O Z_1 O^T\otimes \one) \ket{\Omega_n}.
    \end{equation*}
    Here, we have used that $O^\dagger = O^T$ and $(\mathbbm{1} \otimes O) \ket{\Omega_i} = (O^T \otimes \mathbbm{1}) \ket{\Omega_i}$.
    Now consider an orthogonal unitary $O$ with a lightcone of size less than $n/4$. Note that the internal expectation value is equal to one whenever $i$ lies outside of the lightcone of $O Z_1 O^T$, since then $P^{(i)}$ and $O Z_1 O^T$ commute. Since the expectation value is otherwise no less than $-1$ it follows that the overall expectation value is larger than $1/2$.
    Meanwhile, in a Haar-random orthogonal unitary, the expectation value of the expectation value is exponentially small in $n$.
    This implies that $\lVert \Phi^{(2)}_\mathcal{E} - \Phi^{(2)}_{O_H} \rVert_\lozenge \geq 1/4$, which contradicts $\mathcal{E}$ being an additive error design. The rest of the statement follows. 
\end{proof}

\subsection{Orthogonal state designs: relative error}\label{section: orthogonal rel state}

\begin{theorem}\label{theorem:orthogonal_rel_state}
    Let $\mathcal{E}$ be an ensemble such that $\Phi_\mathcal{E}\left(\ketbra{00}\right)$ is a relative-error $\epsilon$-approximate orthogonal state $2$-design of depth $d$ for $\epsilon < 1$. The following 2 statements hold:
    \begin{enumerate}
                \item If the architecture has diameter $\dia$ and $\mathcal{E}$ is locally independent (Def.~\ref{def:local}), then $d = \Omega(\dia)$.
        \item If $\mathcal{E}$ is a generalized random brickwork (Def.~\ref{def:brickwork}), then $d = \Omega(n)$.
    \end{enumerate}
\end{theorem}

\begin{proof}
    We follow a similar strategy as in the proof of Theorem~\ref{theorem:clifford_rel_state}.
    We know that the state $\ketbra{\Omega}$ lies in the commutant. To prove the theorem we take the trace inner product of the relative-error state design with the psd $(Z_i\otimes \one)\ketbra{\Omega}(Z_i\otimes \one)$.
    For an orthogonal $O$ we get
    \begin{equation}
        \begin{split}
            \tr\left[(Z_i\otimes \one) \ketbra{\Omega_n} (Z_i\otimes \one) O\tn{2}\ketbra{00}O\tnt{2}\right]
            &= \frac{1}{2^n}\bra{00} O\tnt{2} Z_i\tn{2} O\tn{2} \ket{00}.
        \end{split}
    \end{equation}
    For part 1 we compute that 
    \begin{equation}
    \begin{split}
        \frac{1}{2^n}\E_{O\sim \O}\bra{00} O\tnt{2} Z_i\tn{2} O\tn{2} \ket{00}
        &= \frac{1}{2^n}\tr\left[ Z_i\tn{2} \Phi_{\O}(\ketbra{00})\right]\\
        &= \frac{1}{2^{2n}(2^n+1)}\tr\left[Z_i\tn{2} (\one + \F + \Pi_o)\right]\\
        &= \frac{1}{2^{2n}(2^n+1)}\tr\left[\F + \Pi_o\right]\\
        &= \frac{2}{2^n(2^n+2)},
    \end{split}
    \end{equation}
    and similar for the perturbations $Z_j$ and $Z_i Z_j$. We get a contradiction similar as in Eq.~\eqref{eqn: Clifford rel state contradiction}, provided that $\mathcal{E}$ is locally independent.

    For part 2 we note that the real Clifford group is an orthogonal unitary 2-design~\cite{hashagen_real_2018} (and is moreover finite), and that in every layer we thus have a constant probability that the real Clifford factorizes in the first qubit. This implies the lower bound $d = \Omega(n)$ if $\mathcal{E}$ is a generalized random brickwork.
\end{proof}

\subsection{Orthogonal state designs: additive error}\label{section:orthogonal add state}

We prove the following theorem, which shows that the action of the Haar orthogonal twirl and the Haar unitary twirl are equal (up to small additive error) when applied to states that are unentangled between copies.

\begin{theorem} [Equivalence of orthogonal and unitary designs on unentangled states] \label{thm: orthogonal to unitary}
    Consider a state $\rho$ on $\mathcal{H}^{\otimes k}$. Suppose that $\rho$ obeys the PPT condition on each copy of $\mathcal{H}$, i.e.~$\rho^{\Gamma_i} \geq 0$ for all $i = 1,\ldots,k$, where $\Gamma_i$ denotes the partial transpose on the $i$-th copy.
    Then the Haar orthogonal twirl of $\rho$ is equal to the Haar unitary twirl of $\rho$ up to additive error:
\begin{equation}
\lVert \Phi_{\O}(\rho) - \Phi_{\U}(\rho) \rVert_1 \leq \mathcal{O}(k^{7/2}/2^n) + \mathcal{O}(k/2^{n/2}).
\end{equation}
\end{theorem}
Before presenting the proof of the theorem, we briefly discuss its consequences for the construction of low-depth orthogonal state designs. Concretely we have the following corollary:
\begin{corollary}\label{cor:orth add state}
Consider the two-layer superblock brickwork circuit composed of orthogonal $O(\epsilon)$ additive error unitary $k$-designs acting on $\xi = \Omega(\log(nk /\epsilon))$ qubits. This ensemble, with depth $\Or(n\poly(k)\log(1/\epsilon))$ is an $\epsilon$-approximate additive error orthogonal state $k$-design.
\end{corollary}
\begin{proof}
Theorem~\ref{thm: orthogonal to unitary} immediately implies that the two-layer brickwork circuit composed of small orthogonal unitaries acting on $\xi = \Omega(\log(nk /\varepsilon))$ qubits, is an orthogonal state design generator.

In order for a unitary ensemble to be an orthogonal state design generator, we require that it replicate the action of a Haar orthogonal unitary on states of the form $\rho = \dyad{\psi}^{\otimes k}$.
Such states clearly obey the PPT condition in Theorem~\ref{thm: orthogonal to unitary}.
In the context of the two-layer brickwork ensemble, this allows us to replace each small Haar orthogonal unitary with a small Haar unitary up to additive error $\mathcal{O}(n k^{7/2}/2^\xi) + \mathcal{O}(n k/2^{\xi/2})$.
From Theorem~1 of Ref.~\cite{schuster_random_2025}, the two-layer brickwork circuit with small Haar unitaries forms a unitary design with relative error $\mathcal{O}(nk^2/2^\xi)$. This translates to an additive error of the same magnitude. Finally, we can apply Theorem~\ref{thm: orthogonal to unitary} again, to show that the Haar unitary twirl applied to $\rho$ is equal to the Haar orthogonal twirl up to additive error $\mathcal{O}(k^{7/2}/2^n) + \mathcal{O}(k/2^{n/2})$. The total error from this sequence of steps is less than $\varepsilon$ provided that $\xi = \Omega(\log(nk/\varepsilon))$, as claimed. Note that we can construct orthogonal designs in depth $\Or(n\poly(k)\log(1/\epsilon))$ by Ref.~\cite{odonnell_explicit_2023}, so we can replace the superblocks on $\xi$ qubits by linear-depth orthogonal group designs. This all adds up to additive-error orthogonal state designs in depth $\Or(\log(n/\epsilon) \poly(k))$.
\end{proof}
We can now prove Theorem~\ref{thm: orthogonal to unitary}:
\begin{proof}
    We prove that $\lVert \Phi_{\O}(\rho) - \Phi_{\U}(\rho) \rVert_1 \leq \mathcal{O}(k^{7/2}/2^n) + \mathcal{O}(k/2^{n/2})$.
    Our proof proceeds in five steps.
    \begin{enumerate}
        \item The Haar orthogonal twirl squares to itself, $\Phi_{\O} = \Phi_{\O} \circ \Phi_{\O}$.

        \item One can insert projectors onto the distinct subspace after the first and second Haar orthogonal twirls, up to small additive error.
        Specifically, we will show that
        \begin{equation}
            \left\lVert \big[\bPD \circ \Phi_{\O} \circ \bPD \circ \Phi_{\O} \big](\rho) - \big[\Phi_{\O} \circ \Phi_{\O} \big](\rho) \right\rVert_1 \leq 5 k / 2^{n/2},
        \end{equation}
        where $\bPD(\cdot) \equiv \PD (\cdot) \PD$ denotes conjugation by the distinct subspace projector.

        \item Within the distinct subspace, the Haar orthogonal twirl and Haar unitary twirl are equal up to small relative error.
        Specifically, we will show that
        \begin{equation}
            \left(1-\varepsilon' \right) \big[ \bPD \circ \Phi_{\U} \circ \bPD \big]
            \preceq
            \big[ \bPD \circ \Phi_{\O} \circ \bPD \big]
            \preceq
            \left(1+\varepsilon'\right) \big[ \bPD \circ \Phi_{\U} \circ \bPD \big],
        \end{equation}
        with $\varepsilon' = \mathcal{O}(k^{7/2}/2^n)$.
        This allows us to replace the second Haar orthogonal twirl above with a Haar unitary twirl.

        \item One can remove the projectors onto the distinct subspace after the Haar orthogonal and Haar unitary twirls, up to small additive error,
        \begin{equation}
            \left\lVert \big[ \bPD \circ \Phi_{\U} \circ \bPD \circ \Phi_{\O} \big](\rho) - \big[\Phi_{\U} \circ \Phi_{\O} \big](\rho) \right\rVert_1 \leq 5 k / 2^{n/2},
        \end{equation}
        similar to step 2.

        \item The composition of the Haar unitary and Haar orthogonal twirl is equal to the Haar unitary twirl, $\Phi_{\U} \circ \Phi_{\O} = \Phi_{\U}$.
    \end{enumerate}
    In total, these steps show that $\lVert \Phi_{\O}(\rho) - \Phi_{\U}(\rho) \rVert_1 \leq \mathcal{O}(k^{7/2}/2^n) + 10 k/2^{n/2}$, as claimed.

    Let us begin with step 2. From the triangle inequality, we have
    \begin{equation}
    \begin{split}
            \big\lVert \big[\bPD \circ \Phi_{\O} \circ \bPD \circ \Phi_{\O} \big]& (\rho)  - \big[\Phi_{\O} \circ \Phi_{\O} \big](\rho) \big\rVert_1 
            \\
            & \leq  \left\lVert \big[\bPD \circ \Phi_{\O} \circ \big( \bPD - \mathbbm{1} \big) \circ \Phi_{\O} \big](\rho) \right\rVert_1\\
            &+
            \left\lVert \big[ \big( \bPD - \mathbbm{1} \big) \circ \Phi_{\O} \circ \mathbbm{1} \circ \Phi_{\O} \big](\rho) \right\rVert_1.
    \end{split}
    \end{equation}
    To proceed, note that the operation $\bPD \circ \Phi_O$ cannot increase the 1-norm of an operator.
    Thus, the first term above is upper bounded by $\left\lVert \big[ \big( \bPD - \mathbbm{1} \big) \circ \Phi_{\O} \big](\rho) \right\rVert_1$.
    This is in fact equal to the second term, since $\Phi_{\O} \circ \mathbbm{1} \circ \Phi_{\O} = \Phi_{\O}$.
    These give,
    \begin{equation}
            \big\lVert \big[\bPD \circ \Phi_{\O} \circ \bPD \circ \Phi_{\O} \big] (\rho)  - \big[\Phi_{\O} \circ \Phi_{\O} \big](\rho) \big\rVert_1 \leq 2 \left\lVert \big[ \big( \bPD - \mathbbm{1} \big) \circ \Phi_{\O} \big](\rho) \right\rVert_1.
    \end{equation}
    We will now bound the right side.

    From the gentle measurement lemma, the right side is bounded as,
    \begin{equation}
            2 \left\lVert \big[ \big( \bPD - \mathbbm{1} \big) \circ \Phi_{\O} \big](\rho) \right\rVert_1
            \leq 
            4 \sqrt{ 1 - \tr(\PD  \Phi_{\O}(\rho) ) },
    \end{equation}
    where $\tr(\PD  \Phi_{\O}(\rho) )$ is the probability of measuring $\Phi_{\O}(\rho)$ in the distinct subspace.
    We can now apply the operator inequality, $\mathbbm{1}-\PD \leq \sum_{i<j} \Pi_{ij}$, where $\Pi_{ij} = \sum_{x \in \{0,1\}^n} \dyad{x}_i \otimes \dyad{x}_j$ projects onto states with the same bitstring in copy $i$ as in copy $j$.
    This yields, 
    \begin{equation} \label{eq: bound to Piij}
            2 \left\lVert \big[ \big( \bPD - \mathbbm{1} \big) \circ \Phi_{\O} \big](\rho) \right\rVert_1
            \leq 
            4 \sqrt{ \sum_{i<j} \tr(\Pi_{ij}  \Phi_{\O}(\rho) ) }.
    \end{equation}
    We can bound each probability as follows.
    Note that $\tr( \Pi_{ij} \Phi_{\O}(\rho)) = \tr( \Phi_{\O}(\Pi_{ij}) \rho)$, since $\Phi_{\O}^T = \Phi_{\O}$.
    The twirl of $\Pi_{ij}$ over an orthogonal 2-design yields, by Eq.~\eqref{eqn: orthogonal 2 twirl},
    \begin{equation}
        \Phi_{\O}(\Pi_{ij}) = \frac{1}{2^n+2} \left( \mathbbm{1}_{ij} + \F_{ij} + \Pi_{o,ij} \right).
    \end{equation}
    For the first two terms, we have $\tr( \one_{ij} \rho ) = 1$, and $| \tr( \F_{ij} \rho ) | \leq 1$ since $\lVert \F_{ij} \rVert_\infty \leq 1$.
    To bound the third term, we apply the PPT condition on subsystem $i$.
    We have,
    \begin{equation}
        \tr( \Pi_{o,ij} \rho ) = \tr( (\Pi_{o,ij} )^{\Gamma_i} \rho^{\Gamma_i} ) \leq \lVert ( \Pi_{o,ij} )^{\Gamma_i} \rVert_\infty \cdot \lVert \rho^{\Gamma_i} \rVert_1 =  \lVert \F_{ij}  \rVert_\infty \cdot 1 = 1.
    \end{equation}
    In the second step, we apply Holder's inequality.
    In the third step, we use that $(\Pi_{o,ij} )^{\Gamma_i} = \F_{ij}$, and that $\lVert \rho^{\Gamma_i} \rVert_1 = \tr( \rho^{\Gamma_i} ) = \tr(\rho) = 1$, where the first equality follows from the PPT condition, $\rho^{\Gamma_i} \geq 0$.
    In total, these bounds yield
    \begin{equation}
        \tr(\Pi_{ij}  \Phi_{\O}(\rho) ) \leq \frac{3}{2^n+2}.
    \end{equation}
    Inserting into Eq.~\eqref{eq: bound to Piij} and performing the sum over $i<j$, we find
    \begin{equation}
            2 \left\lVert \big[ \big( \bPD - \mathbbm{1} \big) \circ \Phi_{\O} \big](\rho) \right\rVert_1
            \leq 
            4 \sqrt{ \frac{k(k-1)}{2} \frac{3}{2^n+2} } \leq \frac{5 k}{2^{n/2}},
    \end{equation}
    which completes our proof of step 2.
    
    We can prove step 4 in an extremely similar manner.
    Following the same logic as for step 2, we have
    \begin{align} \label{eq: bound PD U O}
            \big\lVert \big[\bPD \circ \Phi_{\U} &\circ \bPD \circ \Phi_{\O} \big] (\rho)  - \big[\Phi_{\U} \circ \Phi_{\O} \big](\rho) \big\rVert_1 \notag\\
            &\leq  \left\lVert \big[ \big( \bPD - \mathbbm{1} \big) \circ \Phi_{\O} \big](\rho) \right\rVert_1
            + 
            \left\lVert \big[ \big( \bPD - \mathbbm{1} \big) \circ \Phi_{\U} \big](\rho) \right\rVert_1.
    \end{align}
    From the previous paragraph, the first term is upper bounded by $(5k/2^{n/2})/2$.
    Moreover, the second term is upper bounded by the same value.
    This follows since $\Phi_U = \Phi_{\O} \circ \Phi_{\U}$, and $\Phi_{\U}(\rho)$ obeys the PPT condition because $\rho$ obeys it.
    Thus, Eq.~\eqref{eq: bound PD U O} is bounded by $5k/2^{n/2}$, as claimed.

    It remains only to prove step 3.
    To show this, we first establish a technical lemma, which allows one to easily bound the relative error \emph{within the distinct subspace} between the twirl over any unitary ensemble $\mathcal{E}$ and the Haar unitary twirl.
    The lemma is closely adapted from Lemma~2 of Ref.~\cite{schuster_random_2025}.
    To more easily bound the error, it leverages the ``approximate Haar unitary twirl''~\cite{schuster_random_2025},
    \begin{equation}
        \Phi_a( \cdot ) = \frac{1}{2^{nk}} \sum_{\pi\in S_k} \tr( (\cdot) \pi^{-1} ) \cdot \pi,
    \end{equation}
    which is close to the exact Haar unitary twirl up to relative error $k^2/2^n$ (Lemma~1 of Ref.~\cite{schuster_random_2025}).
    \begin{lemma}[Relative error on the distinct subspace from EPR states] \label{lemma: relative error EPR}
    Consider a unitary ensemble $\mathcal{E}$ and its twirl $\Phi_{\mathcal{E}}$.
    Within the distinct subspace, the twirl is approximated by $\Phi_U$ up to relative error,
    \begin{equation} \label{eq: relative error EPR}
        \varepsilon' = \frac{4^{n k}}{k!} \left( 1 + \frac{k^2}{2^n} \right) \big\lVert [ \Pi_{\text{\emph{dist}}} \otimes \Pi_{\text{\emph{dist}}}] \, \big[ (\Phi_{\mathcal{E}} - \Phi_a) \otimes \mathbbm{1} \big]( \ketbra{\Omega_n} ) \, [ \Pi_{\text{\emph{dist}}} \otimes \Pi_{\text{\emph{dist}}}] \big\rVert_\infty + \frac{k^2}{2^n}.
    \end{equation}
    That is, we have $(1-\varepsilon') \, \bs{\Pi}_{\text{\emph{dist}}} \circ \Phi_{\U} \circ \bs{\Pi}_{\text{\emph{dist}}} \preceq \bs{\Pi}_{\text{\emph{dist}}} \circ \Phi_{\mathcal{E}} \circ \bs{\Pi}_{\text{\emph{dist}}} \preceq (1+\varepsilon') \, \bs{\Pi}_{\text{\emph{dist}}} \circ \Phi_{\U} \circ \bs{\Pi}_{\text{\emph{dist}}}$, with $\varepsilon'$ as above.
\end{lemma}
\begin{proof}
    The proof follows identically to the proof of Lemma~2 and Lemma~7 in Ref.~\cite{schuster_random_2025}.
    The only difference is the presence of the distinct subspace projector.
    Crucially, in the proof of Lemma~2, this modification does not change the eigenvalue of $[ \Phi_a \otimes \mathbbm{1} ](\ketbra{\Omega_n}) = \frac{1}{4^{nk}} \sum_\pi \pi \otimes \pi$ (which has a flat spectrum), owing to the fact that $\PD \otimes \PD$ and $[ \Phi_a \otimes \mathbbm{1} ](\ketbra{\Omega_n})$ commute, since $\PD$ and $\pi$ commute for all $\pi$.
\end{proof}

We can now apply the lemma to bound the relative error between $\bPD \circ \Phi_{\O} \circ \bPD$ and $\bPD \circ \Phi_U \circ \bPD$.
To do so, we first compute $[\Phi_O \otimes \mathbbm{1}] (\ketbra{\Omega_n})$,
\begin{equation}
    [\Phi_{\O} \otimes \mathbbm{1}] (\ketbra{\Omega_n}) = \frac{1}{2^{nk}} \sum_{\sigma,\tau \in M_k} \Wg_{\O}(\sigma,\tau) [ \tau \otimes \sigma ],
\end{equation}
where $\sigma,\tau$ are summed over the generators of the commutant of the orthogonal group.
Multiplying by the distinct subspace projectors yields,
\begin{equation}
    [\PD \otimes \PD] \, [\Phi_{\O} \otimes \one] (\ketbra{\Omega_n}) \, [\PD \otimes \PD]  = \frac{1}{2^{nk}} \sum_{\pi,\tilde{\pi} \in S_k} \Wg_{\O}(\pi,\tilde{\pi}) \, [ \PD  \pi \, \PD \otimes \PD  \tilde{\pi} \, \PD ],
\end{equation}
where $\pi,\tilde{\pi}$ are summed over the generators of the commutant of the unitary group.
This follows since all other members of the commutant vanish upon multiplication with the projector, i.e.~$\sigma \, \PD = 0$ if $\sigma \notin S_k$.
Taking the difference with $[\PD \otimes \PD] \, [\Phi_a \otimes \mathbbm{1}] (\ketbra{\Omega_n}) \, [\PD \otimes \PD]$ gives,
\begin{equation}
\begin{split}
    [\PD \otimes \PD]& \, [(\Phi_{\O}-\Phi_a) \otimes \mathbbm{1}] (\ketbra{\Omega_n}) \, [\PD \otimes \PD]\\
    &= \frac{1}{2^{nk}} \sum_{\pi,\tilde{\pi} \in S_k} \left( \Wg_{\O}(\pi,\tilde{\pi}) - \frac{1}{2^{nk}} \delta_{\pi,\tilde{\pi}} \right) \, [ \PD  \pi \, \PD \otimes \PD  \tilde{\pi} \, \PD ].
\end{split}
\end{equation}
Applying the triangle inequality, the spectral norm of the above is bounded by
\begin{equation} \label{eq: bound spectral norm orth epr}
    \lVert \ldots \rVert_\infty \leq  \frac{1}{2^{nk}} \sum_{\pi,\tilde{\pi} \in S_k}  \left| \Wg_{\O}(\pi,\tilde{\pi}) - \frac{1}{2^{nk}} \delta_{\pi,\tilde{\pi}}  \right| = \frac{k!}{2^{nk}} \sum_{\pi \in S_k}  \left| \Wg_{\O}( \pi) - \frac{1}{2^{nk}} \delta_{\pi,\one}  \right|,
\end{equation}
where in the third expression we use that the Weingarten function depends only on the difference, $\pi \tilde{\pi}^{-1} \rightarrow \pi$, if $\pi,\tilde\pi$ are permutations and thus invertible.
Separating out the identity term, we have
\begin{equation}
\begin{split}
    \sum_{\pi \in S_k} \left| \Wg_{\O}( \pi) - \frac{1}{2^{nk}} \delta_{\pi,\mathbbm{1}}  \right| 
    & \leq
    \sum_{\pi \in S_k\setminus \{\one\}} \left| \Wg_{\O}( \pi) \right| + \left| \Wg_{\O}( \mathbbm{1} ) - \frac{1}{2^{nk}} \right|.
\end{split}
\end{equation}
By Theorem~4.11 of~\cite{collins2017weingarten}, we deduce that $\left| \Wg_{\O}( \one ) - \frac{1}{2^{nk}} \right| = \frac{1}{2^{nk}} \mathcal{O}\left(\frac{k^{7}}{2^{2n}}\right)$. We prove a lemma for the other term.

\begin{lemma}\label{lemma orthogonal sum weingarten}
    If $2^{n} > 13 k^{7/2}$, then
    \begin{equation}
        \sum_{\pi \in S_k\setminus \{\one\}} \left| \Wg_{\O}( \pi) \right| = \frac{1}{2^{nk}} \Or\left(\frac{k^2}{2^n}\right).
    \end{equation}
\end{lemma}
\begin{proof}[Proof of Lemma~\ref{lemma orthogonal sum weingarten}]
    We know by Theorem~4.11 of~\cite{collins2017weingarten} that for a permutation $\sigma$ we have
    \begin{equation}\label{eqn: Collins}
        |\Wg_{\O}(\sigma)| \leq \frac{1}{2^{n(k + |\sigma|)}}\frac{1}{1- \frac{144k^7}{2^{2n}}} \prod_{i = 1}^{\#\sigma} \frac{(2\mu_i - 2)!}{(\mu_i -1)! \mu_i !},
    \end{equation}
    where $\#\sigma$ is the number of cycles of $\sigma$. We defined $\mu_i$ as the cycle type of $\sigma$ (the sizes of the cycles). We defined $|\sigma| \coloneqq k - \#\sigma$.  We note that $\sum_i \mu_i = k$, while $\sum_i 1 = \#\sigma$, which implies that $\sum_i (\mu_i - 1) = k - \#\sigma = |\sigma|$. We can now compute
    \begin{equation}
    \begin{split}
        \prod_{i=1}^{\#\sigma} \frac{(2\mu_i - 2)!}{(\mu_i -1)! \mu_i !}
        &= \prod_{i=1}^{\#\sigma} \frac{1}{\mu_i}{\binom{2(\mu_i - 1)}{\mu_i - 1}}\\
        &\leq \prod_{i=1}^{\#\sigma} 2^{2(\mu_i - 1)}\\
        &= 2^{2\sum_i (\mu_i - 1)}\\
        &= 2^{2|\sigma|}.
    \end{split}
    \end{equation}
    The second step follows from standard bounds on the central binomial coefficient: $\binom{2m}{m}\leq 2^{2m}$.
    By Eq.~\eqref{eqn: Collins} we then have that
    \begin{equation}
        \sum_{\sigma \in S_k \setminus \{\one\}} |\Wg_{\O}(\sigma)|
        \leq \frac{1}{2^{nk}(1- \frac{144k^7}{2^{2n}})} \sum_{\sigma \in S_k \setminus \{\one\}} 2^{-(n-2)|\sigma|}.
    \end{equation}
    Note that the number of permutations with exactly $j$ cycles is given by the unsigned Stirling numbers $c(k,j)$.
    Moreover, it is well-known that their generating function over $j$ is given by the rising factorial such that we find:
    \begin{equation}
    \begin{split}
        \sum_{\sigma \in S_k \setminus \{\one\}} |\Wg_{\O}(\sigma)|
        \leq& \frac{1}{2^{nk}(1- \frac{144k^7}{2^{2n}})} \left( 2^{-(n-2)k} \sum_{j = 0}^{k} c(k,j) 2^{(n-2)j} - 1\right)\\
        \leq& \frac{1}{2^{nk}(1- \frac{144k^7}{2^{2n}})} \left( \prod_{j=0}^{k-1} \frac{2^{n-2}+j}{2^{n-2}} - 1\right)\\
        =& \frac{1}{2^{nk}(1- \frac{144k^7}{2^{2n}})} \left((1 + k 2^{-n+2})^k - 1\right)\\
        \leq& \frac{1}{2^{nk}(1- \frac{144k^7}{2^{2n}})} \left(e^{k^2 2^{-n+2}} - 1\right)\\
        \leq& \frac{1}{2^{nk}(1- \frac{144k^7}{2^{2n}})} \frac{1}{\log(2)} k^2 2^{-n+2},
    \end{split}
    \end{equation}    
    where in the final step we used that $k^2 2^{-n+2} \leq \log(2)$ such that we can use the inequality $e^x \leq 1+x/\log(2)$. 
    As we assumed that $2^n > 13k^{7/2}$, we see that $\frac{1}{1- \frac{144k^7}{2^{2n}}} = \Or(1)$. 
    We get $\sum_{\sigma \in S_k \setminus \{\one\}} |\Wg_{\O}(\sigma)| = \frac{1}{2^{nk}} \Or\left(\frac{k^2}{2^n}\right)$, which is what we wanted to prove.
\end{proof}

Using Lemma~\ref{lemma orthogonal sum weingarten} we get
\begin{equation}
\begin{split}
    \sum_{\pi \in S_k} \left| \Wg_{\O}( \pi) - \frac{1}{2^{nk}} \delta_{\pi,\mathbbm{1}}  \right|
    \leq \frac{1}{2^{nk}} \cdot \left(\mathcal{O}\left(\frac{k^{7}}{2^{2n}}\right) + \Or\left(\frac{k^2}{2^n}\right)\right) = \frac{1}{2^{nk}}\Or\left(\frac{k^{7/2}}{2^n}\right),
\end{split}
\end{equation}
using that $k^{7/2} < 2^n$. From Eq.~\eqref{eq: bound spectral norm orth epr} and Lemma~\ref{lemma: relative error EPR}, this leads to a relative error $\mathcal{O}\left( \frac{k^{7/2}}{2^{n}}\right)$, as claimed.
This completes our proof of step 3, and hence of Theorem~\ref{thm: orthogonal to unitary}.
\end{proof}

\subsection{Orthogonal anti-concentration}\label{section: orthogonal anti-concentration}

\begin{theorem}\label{theoren:orthogonal_anti-concentrate}
   Let $\mathcal{E}$ be the ensemble of orthogonal matrices drawn from the two-layer superblock architecture. This gives us two layers of $m/2 = n/(2\xi)$ superblocks each, where $\xi = \Or(\log(n/\epsilon))$. This ensemble anti-concentrates for $k=2$, i.e. 
    \begin{equation}
    \mathbb{E}_{O \sim \mathcal{E}}\left|\bra{0} O \ket{0_n}\right|^4 \leq (1 + \epsilon)\ \mathbb{E}_{O \sim \O}\left|\bra{0} O \ket{0_n}\right|^4.
    \end{equation}
\end{theorem}
\begin{proof}
Let $\mathcal{E}_1, \mathcal{E}_2$ be the ensembles consisting of orthogonal matrices sampled from the first and second layer of superblocks respectively. We thus have $\Phi_{\mathcal{E}} = \Phi_{\mathcal{E}_2} \circ \Phi_{\mathcal{E}_1}$.
    We compute that
    \begin{equation}
        \begin{split}
            \mathbb{E}_{O \sim \mathcal{E}}\left|\bra{0^n} O \ket{0^n}\right|^4 &= \tr\left[\Phi_{\mathcal{E}_1}(\ketbra{0^n}^{\otimes2}) \Phi_{\mathcal{E}_2}(\ketbra{0^n}^{\otimes2})\right]\\
            &= \frac{1}{2^{2n}(2^{2\xi} + 2)^{m}}\tr\left[\left(\one + \mathbb{F} + \Pi_o\right)^{\otimes m/2} \left(\one + \mathbb{F} + \Pi_o \right)^{\overrightarrow{\otimes m/2}}\right]\\
            &\leq \frac{1}{2^{2n}(2^{2n} + 2)}\tr\left[\left(\one + \mathbb{F} + \Pi_o\right)^{\otimes m/2} \left(\one + \mathbb{F} + \Pi_o\right)^{\overrightarrow{\otimes m/2}}\right],
        \end{split}
    \end{equation}
    where the arrow indicates that the operators $\one + \mathbb{F} + \Pi_o$ act to the patches $(2,3),\ldots, (m, 1)$. Let us now expand the product as a sum
    \begin{equation}
        \begin{split}
            \tr\left[\left(\one + \mathbb{F} + \Pi_o\right)^{\otimes m/2} \left(\one + \mathbb{F} + \Pi_o\right)^{\overrightarrow{\otimes m/2}}\right] &= \sum_{a,b \in \{\one, \mathbb{F}, \Pi_o\}^{m/2}}\tr[a_1b_{m/2}\otimes a_1b_1 \otimes \cdots \otimes a_{m/2}b_{m/2}],
        \end{split}
    \end{equation}
    where $a_ib_j$ acts on two copies of $\xi$ qubits. We can easily compute that
    \begin{equation}
    \tr[a_ib_j] = 
    \begin{cases}
    2^{2\xi} & \text{if } a_i = b_j \\
    2^{\xi} & \text{if } a_i \neq b_j.
    \end{cases}
    \end{equation}
    This means that a term from the sum only depends on the number of \textit{walls}, i.e.\ the number of $i$ such that $a_i \neq b_i$ plus the number of $i$ such that $a_{i} \neq b_{i+1}$. Let us define $w_{a,b}$ as the number of walls for given $a,b \in \{\one, \mathbb{F}, \Pi_o\}^{m/2}$. We define a \textit{room} as the set of patches in between two walls. We thus get that
    \begin{equation}
        \begin{split}
            \frac{1}{2^{2n}}\sum_{a,b \in \{\one, \mathbb{F}, \Pi_o\}^{m/2}} \tr[a_1b_{m/2}\otimes a_1b_1 \otimes \cdots \otimes a_{m/2}b_{m/2}] = \sum_{a,b \in \{\one, \mathbb{F}, \Pi_o\}^{m/2}} \frac{1}{2^{\xi \cdot w_{a,b}}}.
        \end{split}
    \end{equation}
    To upper bound the expression, we will count the number of configurations $a,b$, that give a number of walls $w$. There are $\binom{2m}{w}$ possible positions where the walls could be placed. Furthermore, inside a room (an interval bounded by two walls), we need to have $a_i = b_j$. This can be achieved in three different ways (i.e.\ $a_i$ equals $\one, \mathbb{F}$ or $\Pi_o$). We thus get a factor of $3^{\min\{1,w\}}$, as $\min\{1,w\}$ is the total number of rooms. As we are only interested in upper bounding the expression, we can ignore the fact that two adjacent rooms need to have a different sign ($\one$, $\F$ or $\Pi_o$). We can also ignore the fact that there exists no configuration with exactly 1 wall. Let us upper bound our equation as
    \begin{equation}
        \begin{split}
            \sum_{a,b \in \{\one, \mathbb{F}, \Pi_o\}^{m/2}} \frac{1}{2^{\xi \cdot w_{a,b}}} 
            &\leq \sum_{w=0}^{2m} 3^{\min\{1,w\}} \binom{2m}{w}\frac{1}{2^{\xi \cdot w}}\\
            &\leq 3\sum_{w=0}^{2m} \binom{2m}{w}\frac{1}{2^{\xi \cdot w - \log_2(3) w}}\\
            &= 3 \left( 1 + \frac{1}{2^{\xi - \log_2(3)}}\right)^{2m}\\
            &\leq 3 \exp\left(\frac{2m}{2^{\xi - \log_2(3)}}\right)\\
            &\leq 3\left(1 + \frac{2m}{\log_e(2) \cdot 2^{\xi - \log_2(3)}}\right).
        \end{split}
    \end{equation}
    In the final step we used that $\exp(x) \leq 1 + x/\log_e(2)$ for $0 \leq x\leq \log_e(2)$.
    We use $\xi = \Or(\log(n/\epsilon))$ to get 
    \begin{equation}
    \frac{3 \cdot 2m}{\log_e(2) \cdot 2^{\xi - \log_2(3)}} \leq \epsilon.
    \end{equation}
    Combining all the things we learned, we find that
    \begin{equation}
    \mathbb{E}_{O \sim \mathcal{E}}\left|\bra{0_n} O \ket{0_n}\right|^4 \leq (1 + \epsilon)\frac{3}{2^n(2^n+2)} = (1 + \epsilon)\ \mathbb{E}_{O \sim \O}\left|\bra{0_n} O \ket{0_n}\right|^4,
    \end{equation}
    which is what we set out to prove.
\end{proof}

We note that Theorem \ref{theoren:orthogonal_anti-concentrate} also follows from the computation for random MPS in Ref.~\cite{sauliere_universality_2025} along the lines of Ref.~\cite{magni_anticoncentration_2025}.
In this way, higher-order anticoncentration can be shown as well.

\section{Unitary symplectic group}
\label{sec:symplectic-group}

\subsection{Preliminaries}
We define a \textit{symplectic form} as a $D\times D$ matrix $J$ such that $J^2 = \one$ and $J^T = J^\dagger = -J$.
Note that for such an operator to exist, the dimension of the underlying space has to be even.
We say that a unitary $U\in\U(D)$ is symplectic if $UJU^T=J$ and denote the group of unitary symplectic matrices as $\USp(D,J)\equiv\USp(D)$.
If the context is clear, we will simply call $\USp(D)$ the symplectic group.

We note that we could make different choices of $J$, but that they are all equal up to a basis change. 
The unitary symplectic group acts transitively on the set of pure quantum states, i.e.~for every pair of states $(\psi,\varphi)$ there is a $U\in\USp(D)$ such that $U\ket\psi = \ket\varphi$ \cite{zimboras_symmetry_2015}.
Clearly, this does not depend on the choice of $J$.
A consequence of transitivity is that $\USp(D)$ induces the same measure as $\U(D)$ on pure quantum states and thus the ensemble $\{U\ket\psi\}_{U\sim\USp(D)}$ is an exact state design of any order, see also Ref.~\cite{west_random_2024}.

Another important note is that the mentioned basis change does not need to be local, so when studying the depth of circuits consisting of local gates, it is important to take into account what $J$ is.
Here, we consider symplectic forms over $n$ qubits, $D=2^n$, of the form $J = J_x \coloneqq i^{|x|} \bigotimes_{j=1}^{n} Y^{x_j}$, where $x \in \{0,1\}^n$ is such that the Hamming weight $|x|$ of $x$ is odd.
For instance, if we choose $x = (100\ldots0)$, then $J_1 = iY_1$ is what is usually considered the standard choice of symplectic form and what has been studied in the context of local circuits before~\cite{garcia-martin_architectures_2024, west_random_2024}. 
In this paper, we will however choose $x = (11\ldots 11)$, and assume that $n$ is odd. We thus define $J = i^{n}\bigotimes_{j=1}^n Y_j$. 

Let $U = U_r \otimes \one_{n-r} $ be a local gate acting on $r$ qubits. For $U$ to be symplectic for our choice of $J$, we need that $U_r Y\tn{r} U_r^T = Y\tn{r}$. This means that $U_r$ can only be symplectic as a unitary on $r$ qubits, if $r$ is odd. It is therefore not natural to look at symplectic circuits with local gates that act on 2 qubits. A more natural approach is to use local gates on 3 qubits. For example, it is known that the 3-qubit gate $U_3 \coloneqq \exp(i (Z_1 \otimes Z_2 \otimes Z_3))$, together with all 1-qubit gates, span the symplectic group~\cite{bremner2004fungible}. The number of qubits one wants to apply local gates on, is however really dependent on the choice of $J$.

We quickly review some symplectic Weingarten calculus. The commutant of the symplectic representation $U \rightarrow U^{\otimes k}$ is a representation of the Brauer algebra, but by a different representation than in the orthogonal case. We can now define the representation $R_{\Sp}$ as

\begin{equation}
R_{\Sp}(\sigma) \coloneqq \left(\prod_{j=1}^{2k}J_j^{a_\sigma(j)}\right)\sum_{i_1,\ldots i_{2k}} \left(\prod_{j=1}^{2k} \delta_{i_j, i_{f_\sigma(j)}}\right) \ketbra{i_1\ldots i_k}{i_{k+1}\ldots i_{2k}}\left(\prod_{j=1}^{2k}J_j^{b_\sigma(j)}\right),
\end{equation}
where we defined $a_\sigma(j) = 1$ if $j < f_\sigma(j) \leq k$ and 0 otherwise. Likewise, $b_\sigma(j) = 1$ if $k < j < f_\sigma(j)$ and 0 otherwise. $J_j$ is the operator $J$ acting on subsystem $j$. $R_{\Sp}$ is a representation of $\mathcal{B}_k(-D)$, with $D$ the dimension.

We note that for $k=2$ the commutant of the symplectic group is given by $\one, \mathbb{F}$ and $\Pi_s \coloneqq (J \otimes \one)\ketbra{\Omega_n}(J \otimes \one)$, where we recall that $\ket{\Omega_n} = \frac{1}{2^{-n/2}} \sum_{i\in \{0,1\}^n} \ket{ii}$.
We compute the Gram and Weingarten matrices
\begin{align}\label{eqn:symplectic 2 twirl}
    G_{\Sp} = D \begin{pmatrix} D & 1 & -1 \\ 1 & D & 1 \\ -1 & 1 & D \end{pmatrix},\ 
    \Wg_{\Sp} = \frac{1}{D(D+1)(D-2)} \begin{pmatrix} D-1 & -1 & 1 \\ -1 & D-1 & -1 \\ 1 & -1 & D-1 \end{pmatrix} \,.
\end{align}

Hence,
\begin{equation}
\begin{split}
    \EE_{U\sim \Sp} U\tn{2} X U\tnd{2}
    &=
    \frac{1}{D(D+1)(D-2)}
    \Big(
    \left[ (D-1)\tr(X) - \tr(X\F) + \tr(X\Pi_s) \right] \one \\
    &+
    \left[ -\tr(X) + (D-1) \tr(X\F) - \tr(X\Pi_s) \right] \F \\
    &+
    \left[ \tr(X) - \tr(X\F) + (D-1) \tr(X\Pi_s) \right] \Pi_s
   \Big).
   \label{eq:orthogonal-twirl}
\end{split}
\end{equation}

\subsection{Symplectic group designs}\label{section: symplectic matrix}

Our lower bound on the depth of additive-error symplectic group designs proceeds in an exactly parallel manner to our lower bound on additive-error orthogonal group designs.
We consider the canonical choice $J = i^{n}\bigotimes_{j=1}^n Y_j$, for which we prove the following theorem.
\begin{theorem}\label{theorem: symplectic rel group design}
Let $\mathcal{E}$ be an approximate symplectic unitary $k$-design with additive error $\varepsilon < 1/4$. Then at least one symplectic unitary must have a lightcone covering a $1/4$'th fraction of the qubits. 
    This implies $\mathcal{E}$ must have circuit depth $d = \Omega(n)$ in $1\mathrm{D}$ circuits and $d = \Omega(\log n)$ in all-to-all connected circuits for any $k \geq 2$.
\end{theorem}

For completeness, we can also consider the alternative choice $J = i Y_1$ introduced in~\cite{garcia-martin_architectures_2024}.
For this choice, all the local gates needed to construct global symplectic gates are all orthogonal, except the gates acting on the first qubit~\cite{garcia-martin_architectures_2024}. This implies that the depth lower bounds for orthogonal group designs  (in Section~\ref{section:matrix_orthogonal}, based on~\cite{schuster_random_2025}) immediately hold for this choice of symplectic group designs as well. 

\begin{proof}
    To prove the theorem, we will show that any symplectic unitary with light-cone smaller than the system size can be distinguished from a Haar-random symplectic unitary using two parallel queries and constant-depth quantum and classical operations.
    
    Let $U$ be a short-depth symplectic unitary and $\ket{J_n} \equiv (J \otimes \mathbbm{1}) \ket{\Omega_n}$ the EPR state acted on by $J$ on the first copy.
    The state $\ket{J_n}$ can be prepared in constant depth since it is a tensor product of $n$ two-qubit stabilizer states, $\ket{J_n} = \ket{J_1}^{\otimes n}$.
    Our distinguishing protocol proceeds as follows.
    We first perturb the $J$-state by a single-qubit Pauli operator $Z_1$ on the first qubit (taking an arbitrary but fixed ordering of the qubits) of the first copy.
    We then apply $U$ to both copies in parallel.
    Finally,we choose $i \in [n]$ uniformly at random and we measure the expectation value of the local $J$-projector, $P_J^{(i)} = \mathbbm{1}^{\otimes 2(n-1)} \otimes \dyad*{J^{(i)}_1}$, where $\ket*{J^{(n)}_1}$ is the $J$-state on the $i$'th qubit of the two copies.

    This procedure yields an expectation value,
    \begin{equation*}
        \mathbb{E}_{i\sim [n]}\bra{\Omega_n}(J Z_1 \otimes \one) (U^\dagger \otimes U^\dagger) P_J^{(i)} (U\otimes U)(Z_1 J \otimes \one) \ket{\Omega_n}
        =
        \mathbb{E}_{i\sim [n]}\bra{\Omega_n}(J U Z_1 U^\dagger \otimes \one)  P_J^{(i)}(U Z_1 U^\dagger J \otimes \one) \ket{\Omega_n}.
    \end{equation*}
    In the second expression, we used $(\mathbbm{1} \otimes U) \ket{\Omega_n} = (U^T \otimes \mathbbm{1}) \ket{\Omega_n}$ and $J U^T = U^\dagger J$ and $U^* J = J U$.
    Note that $P_J^{(i)}$ and $U Z_1 U^T$ commute whenever $i$ is not in the lightcone of $U Z_1 U^T$. The proof now finished in the same manner as Theorem \ref{theorem: orthogonal rel group design}.
\end{proof}

Moreover, we find a linear lower bound for generalized random brickworks.

\begin{theorem}\label{theorem: symplectic rel group design generalized brickwork}
Let $\mathcal{E}$ be an approximate symplectic unitary $k$-design with relative error $\varepsilon < 1$. If $\mathcal{E}$ is a generalized random brickwork, then $d = \Omega(n)$.
\end{theorem}

\begin{proof}
    We prepare our system in the state $(Z_1\otimes \one) \ket{J_n}$ and then apply a symplectic unitary on both systems. The probability that we measure $(\one \otimes J)\ket{0^n}\tn{2}$ is
    \begin{equation}
    \begin{split}
        |\bra{J_n}(Z_1\otimes \one) U\tn{2}(\one \otimes J)\ket{0^n}\tn{2}|^2 = |\bra{0^n}U^\dagger Z_1 U\ket{0^n}|^2
    \end{split}
    \end{equation}
    For $U$ drawn from the symplectic Haar measure, we get
    \begin{equation}
    \begin{split}
        \E_{U\sim \USp} |\bra{J_n}(Z_1\otimes \one) U\tn{2}(\one \otimes J)\ket{0^n}\tn{2}|^2
        &= \frac{\tr[(Z_1\otimes \one)\ketbra{J_n}(Z_1\otimes \one)(2^n\one - 2\F + 2^n \Pi_s)]}{2^n(2^n+1)(2^n-2)}\\
        &= \frac{1}{2^n(2^n+1)}.
    \end{split}
    \end{equation}
    We note that the symplectic Clifford group (i.e. all Clifford gates that are symplectic) form a symplectic $2$-design for any number of qubits (this group is also called the semi-real Clifford group \cite{bolt1961clifford}, and corresponds to case $B$ in Theorem 1.6 in \cite{guralnick2005decompositions} which proves this group is a $3$-design for the unitary symplectic group). In the case that $\mathcal{E}$ is a generalized random brickwork, we can follow the same argument as for the Clifford and orthogonal case. We thus get that
    \begin{equation}
        \E_{U\sim \mathcal{E}}|\bra{0^n}U^\dagger Z_1 U\ket{0^n}|^2 \geq \Omega(p^d),
    \end{equation}
    where $p$ is the probability that a local gate (in the symplectic Clifford group) is not entangling. We get an architecture-independent linear lower bound on the depth of symplectic relative-error designs.
\end{proof}

\subsection{Symplectic state designs: relative error}\label{section: symplectic rel state}

In this section we prove a lower bound on the depth of relative-error state designs with local symplectic gates.
\begin{theorem}\label{theorem:symplectic_rel_state}
    Let $\mathcal{E}$ be an ensemble of local symplectic circuits  with respect to $J = i^{n}\bigotimes_{j=1}^n Y_j$ of depth $d$, such that $\Phi_\mathcal{E}(\ketbra{0^n}^{\otimes 2})$ is a relative-error $\epsilon$-approximate state $2$-design for $\epsilon < 1$. The lightcone starting from any qubit should contain all other qubits. In particular this implies a depth lower bound of  $\Omega(n)$ in $1\mathrm{D}$ architectures and $\Omega(\log n)$ in the all-to-all architecture.
\end{theorem}
Note that the theorem makes a stronger statement than Theorems~\ref{theorem:clifford_rel_state} and~\ref{theorem:orthogonal_rel_state} in the all-to-all connectivity. Also, we do not require the circuit to be locally independent.
\begin{proof}
    Let $\mathcal{E}$ be an ensemble of locally symplectic circuits of depth $d < n$. We define $\ket{J} \coloneqq \left(J \otimes \one_{n\times n}\right) \ket{\Omega_n}$, where $\ket{\Omega_n}$ is the $n$ qubit Bell state.

    The idea of the proof is that we evaluate the symplectic twirl $\Phi_{\Sp}$ and the $\Phi_\mathcal{E}$ on initial state $\ketbra{0^n}\tn{2}$, and we show that they differ when we take the trace inner-product with the psd operator $(Y_1 \otimes \one)\ketbra{J}(Y_1 \otimes \one)$. More precisely, we show the following two claims.
    
    Claim 1: When evaluating the twirl over $\mathcal{E}$ we have that
        \begin{equation}
        \begin{split}
            \tr[(Y_1\otimes \one) \ketbra{J}(Y_1\otimes \one) \Phi_{\mathcal{E}}\left(\ketbra{0^n}\tn{2}\right)] = 0.
        \end{split}
        \end{equation}

    Claim 2: When evaluating the symplectic Haar twirl we get
        \begin{equation}
        \begin{split}
            \tr[(Y_1\otimes \one) \ketbra{J}(Y_1\otimes \one) \Phi_{\Sp}\left(\ketbra{0^n}\tn{2}\right)] > 0.
        \end{split}
        \end{equation}

    These two claims together directly imply that $(1-\epsilon) \Phi_{\Sp} \left(\ketbra{0^n}^{\otimes 2}\right) \preceq \Phi_{\mathcal{E}} \left(\ketbra{0^n}^{\otimes 2}\right)$ cannot be true for $\epsilon < 1$, thus proving the theorem.
    
    \begin{proof}[Proof of Claim 1:]
        We note that $U^{\otimes 2} \ket{J} = \ket{J}$ for every symplectic $U$. For $U\in \mathcal{E}$, we then have $U^{\otimes 2} (Y_j\otimes \one) \ket{J} = \left(U_{l}^{\otimes 2} (Y_j\otimes \one) \ket{J_{{x_{l}}}}\right) \otimes \ket{J_{\overline{x_l}}}$, where $U_l$ is the part of the circuit $U$ that lies in the lightcone of qubit $j$.
        All local gates outside of this lightcone vanish, which leaves us with $\ket{J_{\overline{x_l}}}$ on the qubits outside of the lightcone.
        We can thus compute
        \begin{equation}
        \begin{split}
            &\tr[(Y_1\otimes \one) \ketbra{J}(Y_1\otimes \one) \Phi_{\mathcal{E}}\left(\ketbra{0^n}\tn{2}\right)]\\
            &= \bra{J}(Y_1\otimes \one) \Phi_{\mathcal{E}}\left(\ketbra{0^n}\tn{2}\right)(Y_1\otimes \one) \ket{J}\\
            &= \mathbb{E}_{U\sim \mathcal{E}} \left[\bra{0^{n}}\tn{2} \left(\left(U_{l}\tn{2} (Y_1\otimes \one) \ketbra{J_{x_l}}(Y_1\otimes \one) U_{l}\tnd{2} \right) \otimes \ketbra{J_{\overline{x_l}}}\right) \ket{0^{n}}\tn{2}\right]\\
            &= 0.
        \end{split}
        \end{equation}
        In the final step, we used that $\left|\braket{00}{J_1}\right| = \frac{1}{\sqrt2}\left|\bra{00} (\ket{10} - \ket{01})\right| = 0$.
          \phantom\qedhere
    \end{proof}

    \begin{proof}[Proof of Claim 2:]
        Uniformly random ensembles of symplectic and unitary states are indistinguishable~\cite{west_random_2024}. We can thus replace the symplectic twirl with the unitary twirl:
        \begin{equation}
        \Phi_{\Sp}\left(\ketbra{0^n}^{\otimes 2}\right) = \Phi_{\U}\left(\ketbra{0^n}^{\otimes 2}\right) = \frac{2}{2^n(2^n+1)} P_{\sym},
        \end{equation}
        where $P_{\sym} = \frac12 (\one + \mathbb{F})$ is the projector onto the symmetric subspace.
        Note that $Y_1J$ is a symmetric operator, i.e.\ $Y_1J = (Y_1J)^{T}$. This means that 
        \begin{equation}
        P_{\sym}(Y_1\otimes \one) \ket{J} = P_{\sym} (Y_1J \otimes \one)\ket{\Omega_n} = (Y_1\otimes \one) \ket{J}.
        \end{equation}
        Hence       
        \begin{equation}
        \begin{split}
            \tr[(Y_1\otimes \one) \ketbra{J}(Y_1\otimes \one) \Phi_{\Sp}\left(\ketbra{0^n}\tn{2}\right)]
            &= \bra{J}(Y_1\otimes \one) \Phi_{\Sp}\left(\ketbra{0^n}^{\otimes 2}\right)(Y_1\otimes \one) \ket{J}\\
            &= \bra{J}(Y_1\otimes \one) \left(\frac{2}{2^n(2^n+1)} P_{\sym}\right) (Y_1\otimes \one) \ket{J}\\
            &= \frac{2}{2^{n}(2^n+1)} > 0,
        \end{split}
        \end{equation}
            \phantom\qedhere
    \end{proof}
    \noindent which is what we set out to prove.
\end{proof}
\subsection{Symplectic state designs: additive error} \label{section: symplectic add state}
In this section we prove that we can construct additive-error state designs in $\Or(\log(n))$ depth, using local symplectic gates. We use a similar approach as in the orthogonal case in Section~\ref{section:orthogonal add state}.
\begin{theorem} [Equivalence of symplectic and unitary designs on unentangled states] \label{thm: symplectic to unitary}
    Consider a state $\rho$ on $\mathcal{H}^{\otimes k} \otimes \mathcal{A}$ that obeys the PPT condition on each copy of $\mathcal{H}$. Then the Haar symplectic twirl of $\rho$ is equal to the Haar unitary twirl of $\rho$ up to additive error
    \begin{equation}
        \lVert \Phi_{\Sp}(\rho) - \Phi_{\U}(\rho) \rVert_1 \leq \mathcal{O}(k^{7/2}/2^n) + \mathcal{O}(k/2^{n/2}).
    \end{equation}
\end{theorem}
Just as in the orthogonal case, the theorem directly implies that we can construct additive-error symplectic state designs in depth~$\Or(\log(kn/\epsilon))$.
\begin{proof}
We define the \textit{symplectic distinct subspace} as the space spanned by strings $(x_1,\ldots, x_k)$ such that $\dyad{x_i} \neq J \dyad{x_j} J$ for all $i\neq j$. Let us define $\SbPD$ as the projector onto this space. The proof of the theorem follows similar steps as the proof of Theorem~\ref{thm: orthogonal to unitary}, except that we project onto the symplectic distinct subspace.

\begin{enumerate}
        \item The Haar symplectic twirl squares to itself, $\Phi_{\Sp} = \Phi_{\Sp} \circ \Phi_{\Sp}$.
        
        \item We prove that
        \begin{equation}
            \left\lVert \big[\SbPD \circ \Phi_{\Sp} \circ \SbPD \circ \Phi_{\Sp} \big](\rho) - \Phi_{\Sp} \circ \Phi_{\Sp} (\rho) \right\rVert_1 \leq 6 k / 2^{n/2},
        \end{equation}
        where $\SbPD(\cdot) \equiv \SbPD (\cdot) \SbPD$ denotes conjugation by the symplectic distinct subspace projector.

        \item Within the symplectic distinct subspace, the Haar symplectic twirl and Haar unitary twirl are equal up to small relative error,
        \begin{equation}
            \left(1-\varepsilon' \right) \big[ \SbPD \circ \Phi_{\U} \circ \SbPD \big]
            \preceq
            \big[ \SbPD \circ \Phi_{\Sp} \circ \SbPD \big]
            \preceq
            \left(1+\varepsilon'\right) \big[ \SbPD \circ \Phi_{\U} \circ \SbPD \big],
        \end{equation}
        with $\varepsilon' = \mathcal{O}(k^{7/2}/2^n)$.

        \item One can remove the projectors onto the symplectic distinct subspace after the Haar symplectic and Haar unitary twirls, up to small additive error,
        \begin{equation}
            \left\lVert \big[ \SbPD \circ \Phi_{\U} \circ \SbPD \circ \Phi_{\Sp} \big](\rho) - \Phi_{\U} \circ \Phi_{\Sp}(\rho) \right\rVert_1 \leq 6 k / 2^{n/2}.
        \end{equation}

        \item The composition of the Haar unitary and Haar symplectic twirl is equal to the Haar unitary twirl, $\Phi_{\U} \circ \Phi_{\Sp} = \Phi_{\U}$.
    \end{enumerate}

We start by proving the symplectic version of step 2. Similarly as in the orthogonal case, we can compute that
\begin{equation}
        \big\lVert \big[\SbPD \circ \Phi_{\Sp} \circ \SbPD \circ \Phi_{\Sp} \big] (\rho)  - \big[\Phi_{\Sp} \circ \Phi_{\Sp} \big](\rho) \big\rVert_1 
        \leq 4 \sqrt{1 - \tr( \Phi_{\Sp}(\rho)\SbPD) }.
\end{equation}

We note that $\one-\SbPD \leq \sum_{i<j} \widehat{\Pi}_{ij}$, where $\widehat{\Pi}_{ij} = \sum_{x \in \{0,1\}^n} \dyad{x}_i \otimes (J\dyad{x}J)_j$. It thus follows that 
\begin{equation}
    4 \sqrt{1 - \tr( \Phi_{\Sp}(\rho)\SbPD) }
    \leq 4 \sqrt{ \sum_{i<j} \tr(\widehat{\Pi}_{ij} \Phi_{\Sp}(\rho) ) }
    = 4 \sqrt{ \sum_{i<j} \tr(\Phi_{\Sp}(\widehat{\Pi}_{ij}) \rho)  }. 
\end{equation}

We know that $\{\one_{ij}, S_{ij}, \Pi_{s,ij}\}$ is a basis of the commutant of the symplectic group over the subsystem $i,j$. We note that $\tr(\Pi_{s,ij} \widehat{\Pi}_{ij}) = 2^n$. Using Eq.~\eqref{eqn:symplectic 2 twirl} we compute that
\begin{equation}
\begin{split}
    \Phi_{\Sp}(\widehat{\Pi}_{ij}) &= \frac{2^n-1}{(2^n + 1)(2^n - 2)}(\one_{ij} + \F_{ij} + \Pi_{s,ij}).
\end{split}
\end{equation}

We can use the PPT condition to bound the last term,

\begin{equation}
    \begin{split}
        \tr( \Pi_{s,ij} \rho )
        &= \tr((\mathbbm{1}\otimes J)_{ij} \Pi_{o,ij} (\mathbbm{1}\otimes J)_{ij} \rho )\\
        &= \tr( ( \Pi_{o,ij} )^{\Gamma_i} (\mathbbm{1}\otimes J)_{ij}\rho^{\Gamma_i}(\mathbbm{1}\otimes J)_{ij} )\\
        &\leq \lVert ( \Pi_{o,ij} )^{\Gamma_i} \rVert_\infty \cdot \lVert (\mathbbm{1}\otimes J)_{ij}\rho^{\Gamma_i}(\mathbbm{1}\otimes J)_{ij} \rVert_1\\
        &= \lVert \F_{ij}  \rVert_\infty \cdot 1\\
        &= 1.
    \end{split}
\end{equation}

This all adds up to

\begin{equation}
    \tr(\Phi_{\Sp}(\widehat{\Pi}_{ij}) (\rho) ) \leq \frac{3(2^n-1)}{(2^n + 1)(2^n - 2)} \leq \frac{4}{2^n}.
\end{equation}

Combining the things we learned gives
\begin{equation}
\big\lVert \big[\SbPD \circ \Phi_{\Sp} \circ \SbPD \circ \Phi_{\Sp} \big] (\rho)  - \big[\Phi_{\Sp} \circ \Phi_{\Sp} \big](\rho) \big\rVert_1 \leq 4\sqrt{\frac{k(k-1)}{2} \frac{4}{2^n}} \leq 6\frac{k}{2^{n/2}}.
\end{equation}
Step 4 is done similarly.

To prove step 3 we see that
\begin{equation}
    [\Phi_{\Sp} \otimes \mathbbm{1}] (\ketbra{\Omega_n}) = \frac{1}{2^{nk}} \sum_{\sigma,\tau \in M_k} \Wg_{\Sp}(\sigma,\tau) [ \tau \otimes \sigma ],
\end{equation}
where $\sigma,\tau$ are summed over the generators of the commutant of the symplectic group.

Multiplying by the symplectic distinct subspace projectors yields
\begin{equation}
    [\SbPD \otimes \SbPD] \, [\Phi_{\Sp} \otimes \mathbbm{1}] (\ketbra{\Omega_n}) \, [\SbPD \otimes \SbPD]  = \frac{1}{2^{nk}} \sum_{\pi,\tilde{\pi} \in S_k} \Wg_{\Sp}(\pi,\tilde{\pi}) \, [ \SbPD  \pi \, \SbPD \otimes \SbPD  \tilde{\pi} \, \SbPD ].
\end{equation}
Note that $\pi,\tilde{\pi}$ sum over $S_k$. This follows since the elements of the commutant that are not a permutation vanish when multiplied with the projector. Similarly as in the orthogonal case we now get that
\begin{equation}
\lVert [\SbPD \otimes \SbPD] \, [(\Phi_{\Sp}-\Phi_a) \otimes \mathbbm{1}] (\ketbra{\Omega_n}) \, [\SbPD \otimes \SbPD] \rVert_\infty \leq \sum_{\pi \in S_k \setminus \one} \left| \Wg_{\Sp}( \pi) \right| + \left| \Wg_{\Sp}( \mathbbm{1} ) - \frac{1}{2^{nk}} \right|.
\end{equation}
We use that $|\Wg_{\Sp}(\mathbbm{1}) - \frac{1}{2^{nk}}| \leq \frac{1}{2^{nk}}\mathcal{O}\left(\frac{k^{7/2}}{2^{2n}}\right)$, which follows from Theorem~4.10 in~\cite{collins2017weingarten}. An argument similar as the proof of Lemma~\ref{lemma orthogonal sum weingarten} tells us that $\sum_{\pi \in S_k \setminus \one} \left| \Wg_{\Sp}( \pi) \right| \leq \frac{1}{2^{nk}}\Or\left(\frac{k}{2^{n}}\right)$. We get that
\begin{equation}
\lVert [\SbPD \otimes \SbPD] \, [(\Phi_{\Sp}-\Phi_a) \otimes \mathbbm{1}] (\ketbra{\Omega_n}) \, [\SbPD \otimes \SbPD] \rVert_\infty \leq \mathcal{O}(k^{7/2}/2^n).
\end{equation}
An argument similar to the proof of Lemma~\ref{lemma: relative error EPR} concludes the proof of step 3.
\end{proof}

\subsection{Symplectic anti-concentration}\label{section: symplectic anti-concentration}
We know that the symplectic twirl equals the unitary twirl on pure states. This directly implies that the gluing lemma~\cite{schuster_random_2025} also works for the symplectic twirl, provided that the input and the output are both pure states, that factorize in the qubit-direction. We thus have anti-concentration in $\mathcal{O}(\log(n))$ depth for locally symplectic circuits.

\begin{theorem}\label{theorem:symplectic_anti-concentrate}
Let $\mathcal{E}$ be the ensemble of symplectic matrices drawn from the two-layer superblock architecture. This gives us two layers of $m/2 = n/(2\xi)$ superblocks each, where $\xi = \Or(\log(k^2n/\epsilon))$. This ensemble anti-concentrates, i.e. 
    \begin{equation}
    \mathbb{E}_{U \sim \mathcal{E}}\left|\bra{0} U \ket{0_n}\right|^{2k} \leq (1 + \epsilon)\ \mathbb{E}_{U \sim \USp}\left|\bra{0} U \ket{0_n}\right|^{2k}.
    \end{equation}
\end{theorem}
\begin{proof}
    Let $\mathcal{E}$ be the ensemble of superblocks with the superblocks drawn from the unitary symplectic group. We define $\mathcal{E}_{1}, \mathcal{E}_{2}$ as the ensemble of superblocks in the first and second layer respectively. We have that
    \begin{equation}\label{eqn:symplectic superblocks anti-concentration}
    \begin{split}
        \E_{U\sim\mathcal{E}} |\bra{0^n} U \ket{0^n}|^{2k} = \tr\left[\Phi_{\mathcal{E}_{1}}(\ketbra{0^n}\tn{k}) \Phi_{\mathcal{E}_{2}}(\ketbra{0^n}\tn{k})\right].
    \end{split}
    \end{equation}
    As the superblocks are locally drawn independently from Haar random symplectic unitaries, and symplectic unitaries are indistinguishable from unitaries on pure states~\cite{west_random_2024}, we can replace the Haar random symplectic unitaries in Eq.~\eqref{eqn:symplectic superblocks anti-concentration} with Haar random unitaries. By the gluing lemma, we now have that 
    \begin{equation}
    \begin{split}
        \E_{U\sim\mathcal{E}} |\bra{0^n} U \ket{0^n}|^{2k}
        &\leq (1+\epsilon)\E_{U\sim \U} |\bra{0^n} U \ket{0^n}|^{2k}\\
        &= (1+\epsilon)\E_{U\sim\USp} |\bra{0^n} U \ket{0^n}|^{2k},
    \end{split}
    \end{equation}
    if the superblocks act on $\Or\left(\log(nk^2/\epsilon)\right)$ qubits.
\end{proof}

\section{Matchgate group}
In this section, we prove our results on the $n$-qubit matchgate group $\mathcal{M}_{n}$, which models the evolution of free fermions (without particle number preservation). 
Matchgate designs have been applied to shadow estimation of fermionic observables \cite{wan2023matchgate}
, scalable benchmarking of quantum computers \cite{helsen2011matchgate}
, and the development of quantum advantage schemes \cite{oszmaniec2022fermion}. 
Here we will prove that the superblock construction for matchgate designs requires depth  $\Omega(n^{1/3})$ to anti-concentrate (though this is probably not sufficient). We will furthermore provide stronger lightcone-based lower bounds for relative and additive-error matchgate state designs.
\subsection{Preliminaries}
The matchgate group $\M(n)$ is generated by $2$-qubit unitaries of the form $e^{\theta P_i Q_{i+1}}$ where $P,Q\in \{X,Y\}$, supplemented with single qubit $X$ gates, where the qubits are laid out on a line. In this specific geometry these generators conspire to generate a small subgroup of the unitary group with a rich structure. To discuss this structure further we need to introduce the so called Majorana operators. \\

The Majorana operators (on $n$ fermionic modes) are a set of operators $\gamma_i$ for $i \in [2n] $ with the properties that $\gamma_i^2 = I$ and $\gamma_i \gamma_j = - \gamma_j\gamma_i$ for $i \neq j$. They form an algebra that can be represented faithfully on the space $(\mathbb{C}^2)\tn{n}$ of qubit states using Hermitian Pauli operators. There are several ways of doing this, but a convenient way is the so called Jordan-Wigner (JW) isomorphism, which maps
\begin{equation}
\begin{aligned}
\gamma_{2i-1}\to Z_1\ldots Z_{i-1} X_i,\\
\gamma_{2i}\to Z_1\ldots Z_{i-1} Y_i.
\end{aligned}
\end{equation}
It turns out that any matchgate circuit effects an orthogonal transformation at the level of Majoranas; we can write a matchgate unitary as $U(O)$ with $O\in \O(2n)$ an orthogonal operator, which acts on Majorana operators as
\begin{equation}
U(O) \gamma_i U(O)\ct = \sum_{j\in [n]}O_{ij} \gamma_j.
\end{equation}
Through the JW isomorphism, we can thus think of the matchgate group as irreducible projective representation of the Lie group $\O(2n)$ on $(\mathbb{C}^2)\tn{n}$ (sometimes called the Pin representation)\footnote{Often the name matchgate group is reserved for the subgroup $\SO(2n)$. The results we discuss here hold for this group as well with minor alterations to the proofs}. The $k=2$ diagonal representation $(U(O) \cdot U(O)\ct)\tn{2}$ was evaluated in \cite{helsen2011matchgate}, where it was shown that the average over the twofold tensor product projects onto a space spanned by the operators 
\begin{equation}
V_k = \binom{2n}{k}^{-1/2} \sum_{S\subset [2n],\;\; |S|=k}2^{-n}{\gamma_S}\tn{2},
\end{equation}
with $k\in [0:2n]$ and
where the Majorana operator products
\begin{equation}
\gamma_S = \prod_{i\in S} \gamma_i,
\end{equation}
are normally ordered (that is, in increasing order in the set $S\subset [2n]$). Note that the operators $V_k$ are orthonormal under the trace inner product. 
\\

For the purposes of this paper it is also important to evaluate averages over random matchgates that only act on a subset of the qubits. It is important to note that such embedded matchgates only exist if the subset is a sub-interval $[a:b]$ of the interval $[n]$ (with $a< b \in [n]$). These matchgates form a Pin representation of the orthogonal group $O(2(a-b))$ acting on the $a$'th to $b$'th qubit. The second moment average can again be characterized; it has support on the operators $A\otimes V_k^{(a,b)}\otimes B$ where $A$ (resp. $B$) is an arbitrary operator on $(\mathbb{C}^2)\tn{a-1}$ (resp. $(\mathbb{C}^2)\tn{n-b}$), and $V_k^{(a,b)}$
is of the form
\begin{equation}
V_k^{(a,b)} = I\tn{a-1}\otimes \binom{2(a-b)}{k}^{-1/2} \sum_{S\subset [2(a-b)],\;\; |S|=k}2^{-(a-b)}{\gamma^{(a,b)}_S}\tn{2} \otimes I\tn{n-b},
\end{equation}
where $\gamma^{(a,b)}_i$ is the operator obtained by taking the JW representation of $\gamma_i$ and replacing the first $a-1$ $Z$ operators with identity $I$ operators. Note that if $k$ is even this replacement does not change the Majorana product operators i.e. if $S\subset [2(a-b)]$ and $|S|$ even, then
\begin{equation}
\prod_{i\in S}\gamma_{i}^{(a,b)} = \prod_{i\in S}\gamma_{i}.
\end{equation}

\subsection{Matchgate anti-concentration}
We prove the following theorem, which constrains how much superblock matchgate circuits anti-concentrate in the generalized sense (i.e. they approximate the value for uniformly random matchgates up to a constant factor).

\begin{theorem}\label{theorem:matchgate_anti-concentrate}
Consider the ensemble of matchgate random circuits in a $1$D superblock construction with block-size $2\xi$, with uniformly random matchgates inside the blocks (and thus total depth $4\xi$). If this ensemble anti-concentrates, that is
\begin{equation}
2^{2n}\mathbb{E}_{U(O)\sim\mathcal{E}}|\bra{0} U(O)\ket{0}|^{4} \leq (1+\epsilon) 2^{2n}\int dO|\bra{0} U(O)\ket{0}|^{4},
\end{equation}
for some constant $\epsilon$, then $\xi = \Omega(n^{1/3}/\log(n))$.
\end{theorem}
\begin{proof}
We begin by computing the anti-concentration value of random $n$-qubit matchgates, i.e.,
\begin{equation}
\chi = 2^{2n}\int_{SO(2n)} dO|\bra{0}U(O)\ket{0}|^4.
\end{equation}
This is a second moment quantity, easily computed using the 
orthonormal basis $\{V_k\}$ described above. Concretely, from the JW isomorphism it is easy to see that $\dens{0}\tn{2}$ is trace orthogonal to $V_k$ with $k$ odd and that for even $k$ we have
\begin{equation}
\tr\big(\dens{0}\tn{2} V_k\big)^2  = 2^{-2n}\binom{n}{k/2}^2 \binom{2n}{k}^{-1}.
\end{equation}
This follows by noting that through JW, only Majorana strings composed of $\gamma_{2i-1}\gamma_{2i}$ products (i.e. $Z_i$ operators) have non-zero expectation value on $\dens{0}$.
It’s not hard to see (using standard bounds for the binomial) that
\begin{equation}
C_1 \sqrt{n /((2n-k)k)} \leq \binom{n}{k}^2 \binom{2n}{k}^{-1} \leq C_2 \sqrt{n /((2n-k)k)},
\end{equation}
for $k\in [2:2n-2]$, for some constants $C_1,C_2$. From this we immediately obtain that (for some constant $C$):
\begin{equation}
\chi  =2^{2n} \sum_{k\in  [0:2:2n]} \tr\big(\dens{0}\tn{2} V_k\big)^2 = 2+ \sum_{k\in  [2:2:2n-2]} C_2\sqrt{n /((2n-k)k)} \leq C\sqrt{n}.
\end{equation}
We will now compute the same quantity for a $1$D superblock circuit of random matchgates. We divide $n$ qubits into $m$ equal patches of size $n/m = \xi$, and have random matchgates $U(O)_{(i)}$ act on the patches $(i,i+1)$ (assuming periodic boundary conditions for convenience). Using the embedded Majorana notation defined above  we can write
\begin{equation}
\int_{O_i} I\otimes U(O_{i})\tn{2} \dens{0}\tn{2}{U(O_{i})\tn{2}}\ct \otimes I = \sum_{k = 0,\;k \mathrm{ even}}^{4\xi}\frac{2^{-\xi}\binom{2\xi}{k/2}}{\binom{4\xi}{k}^{1/2}} I\otimes V^{(\xi i,\xi (i+1))}_k\otimes I.
\end{equation}
With this we can partially evaluate the anti-concentration value as
\begin{align}
\chi &= 2^{2n} \int_{O_{1},\ldots O_{n}\in O(4\xi)} |\bra{0}\big(\bigotimes_{i = 1}^{m/2}U(O_{2i-1})\bigotimes_{i = 1}^{m/2}U(O_{2i})\ket{0}|^4 \\
&= 2^{2n} \sum_{\substack{k_1, \ldots k_{m/2} \in[0:2:4\xi]\\ l_1, \ldots l_{m/2} \in[0:2:4\xi] }} \prod_{i=1}^{m/2} \frac{2^{-2\xi} \binom{2\xi}{k_i/2}}{\binom{4\xi}{k_i}}\tr\big( (V^{(1,2\xi)}_{k_1}\otimes \ldots \otimes V^{(n-2\xi+1,n)}_{k_{m/2}})\\&\hspace{10em}\times (V^{(\xi+1,2\xi)}_{l_1}\otimes \ldots \otimes V^{(n-\xi+1,\xi)}_{l_{m/2}})\big)\prod_{i=1}^{m/2} \frac{2^{-2\xi} \binom{2\xi}{l_i/2}}{\binom{4\xi}{l_i}}.
\end{align}
This is a rather complicated equation, but because of its $1\mathrm{D}$ nature and periodic boundary conditions we can write it as the trace of an $m/2$-fold power of a matrix:
\begin{equation}
\chi_\xi = 2^{2n}\tr(M^{m/2}),
\end{equation}
where (with $0/0 = 0$ for convenience):
\begin{align}
M_{vw}  &= \sum_{\substack{k\in [0:2:4\xi]\\l\in [0:2:4\xi]}}\frac{2^{-2\xi}\binom{2\xi}{k/2}}{\binom{4\xi}{k}}\frac{2^{-2\xi}\binom{2\xi}{l/2}}{\binom{4\xi}{l}}\delta_{k-v,l-w} \tr\big((V_v^{(1,\xi)}\otimes V_{k}^{(\xi+1,3\xi)}) (V_{l}^{(1,2\xi)}\otimes V_w^{(2\xi+1,3\xi)})\big)\\
& = \sum_{k\in [0:2:4\xi]}2^{-4\xi}\frac{\binom{2\xi}{k/2}}{\binom{4\xi}{k}}\frac{\binom{2\xi}{(k-v+w)/2}}{\binom{4\xi}{k-v+w}}\frac{\binom{2\xi}{k-v}}{\binom{2\xi}{v}^{1/2}\binom{2\xi}{w}^{1/2}}.
\end{align}
It remains to provide a lower bound for the largest eigenvalue of this matrix. As it is a point-wise positive matrix, any element on the diagonal provides a lower bound for the largest eigenvalue. In particular we can calculate the entry $M_{00}$
\begin{equation}
M_{00} = 2^{-4\xi}\sum_{k\in [0:2:2\xi]}\frac{\binom{2\xi}{k/2} \binom{2\xi}{k/2}}{\binom{4\xi}{k}\binom{4\xi}{k}}\frac{\binom{2\xi}{k}}{\binom{2\xi}{k}}\geq  2^{-4\xi}\bigg( 1 + \binom{2\xi}{1}^2 \binom{4\xi}{2}^{-2}\bigg) = 2^{-4\xi} \big(1 + 2^{-7}\xi^{-2}\big),
\end{equation}
where we evaluated only the first two terms in the sum (this is a lower bound as all entries are non-negative). It follows that 
\begin{equation}
\chi_{\xi} \geq 2^{2n} 2^{-m\xi/2}\big(1 + 2^{-7}\xi\big)^{m/2} = \big(1 + 2^{-7}\xi^{-2}\big)^{n/2\xi}.
\end{equation}
In order to obtain anti-concentration, we require
\begin{equation}
\big(1 + 2^{-7}\xi^{-2}\big)^{n/2\xi} \leq (1+\epsilon)C \sqrt{n},
\end{equation}
which directly leads to $\xi =\Omega(n^{1/3}/\log(n) )$ , which in turn implies the same scaling for the circuit depth.
\end{proof}

\subsection{Matchgate state designs}
We now prove lower bounds for matchgate state designs. We note that due to the explicit one dimensional nature of the matchgate group these results are simpler to state and prove than the analogous statements for the other groups we have discussed. This is so because limiting the depth of a $1\mathrm{D}$ circuit implies the existence of lightcone that does not depend on the internal structure of the circuit. Hence we have no need for extra assumptions to make our proofs work. 
\begin{theorem}\label{thm: matchgate state designs}
Let $\mathcal{E}$ be an ensemble of matchgate circuits of depth $d$. The following now holds:
\begin{enumerate}
\item If $\Phi_{\mathcal{E}}(\dens{0}\tn{2})$ is an additive-error (with $\epsilon<1/(2n-1) $) state design, then $d =\Omega(n)$
\item If $\Phi_{\mathcal{E}}(\dens{0}\tn{2})$ is a relative-error (with $\epsilon<1/(4n-2) $) state design, then $d =\Omega(n)$.
\end{enumerate}
\end{theorem}
\begin{proof}
Consider the Majorana operators $\gamma_1, \gamma_{2n}$. For uniformly random matchgates we have (using the definition of $V_k$):
\begin{equation}
\int_{O(2n)} dO |\bra{0}U(O) \gamma_1 \gamma_{2n} U(O)\ct\ket{0}|^2 = \frac{1}{2n-1}.
\end{equation}
On the other hand, if $U(O)$ is a matchgate circuit of depth $d\leq n/4$, then clearly
\begin{equation}
|\bra{0}U(O) \gamma_1 \gamma_{2n} U(O)\ct\ket{0}|^2  =|\bra{0} \sum_{i=1}^{n/2} \alpha_i \gamma_i \sum_{j=3n/2}^{2n} \beta_j \gamma_j\ket{0}|^2 = 0,
\end{equation}
since $\bra{0}\gamma_i\gamma_j\ket{0} = 0$ unless $i=j$ or $i$ is odd and $j= i+1$. 
Now since the operator $(\gamma_1\gamma_{2n})\tn{2}$ is hermitian and of operator norm $1$, one clearly requires $d\geq n/4$ to obtain any additive error state design of error smaller than $1/(2n-1)$.  Since the existence of relative-error state designs with error $\epsilon$ imply the existence of additive-error state designs with error $2\epsilon$ the same conclusion holds for relative-error designs.
\end{proof}

\addcontentsline{toc}{section}{References}

\printbibliography

\end{document}